%% file: main.tex
\documentclass[runningheads]{llncs}

\usepackage[T1]{fontenc}
\usepackage[USenglish]{babel}
\usepackage[utf8]{inputenc}

\usepackage{amsmath}
\usepackage{amsfonts}
\usepackage{stmaryrd}

\usepackage{color}
\usepackage{xcolor}
\usepackage{graphicx}

\usepackage[ruled,vlined,linesnumbered]{algorithm2e}

\usepackage{mathtools}
\usepackage{mathabx}

\usepackage{paralist}

\usepackage{hyperref}
\usepackage[capitalize]{cleveref}

\usepackage{tikz}
\usetikzlibrary{arrows, arrows.meta, automata, positioning, calc, shapes}

\usepackage{cite}

\usepackage{thm-restate}

\usepackage{multirow}
\usepackage{wrapfig}

\usepackage{hyperref}

\input{./macros}

\begin{document}

\title{Localized Attractor Computations for Infinite-State Games (Full Version)\thanks{Authors are ordered randomly, denoted by
\textcircled{r}. The publicly verifiable record of the randomization is available at \href{https://www.aeaweb.org/journals/policies/random-author-order/search?RandomAuthorsSearch\%5Bsearch\%5D=fKy1kA2NiEmL}{www.aeaweb.org}.
}}

\titlerunning{Localized Attractor Computations for Infinite-State Games (Full Version)}

\author{Anne-Kathrin Schmuck\inst{1}\orcidID{0000-0003-2801-639X} \textcircled{r}
Philippe Heim\inst{2}\orcidID{0000-0002-5433-8133} \textcircled{r}
Rayna Dimitrova\inst{2}\orcidID{0009-0006-2494-8690} \textcircled{r}
Satya Prakash Nayak\inst{1}\orcidID{0000-0002-4407-8681}
}

\authorrunning{A-K. Schmuck \textcircled{r} P. Heim \textcircled{r} R. Dimitrova \textcircled{r} S. P. Nayak}

\institute{Max Planck Institute for Software Systems (MPI-SWS), Kaiserslautern, Germany 
\email{\{akschmuck, sanayak\}@mpi-sws.org}
\and CISPA Helmholtz Center for Information Security, Saarbr\"ucken, Germany
\email{\{philippe.heim, dimitrova\}@cispa.de}
}

\maketitle

\begin{abstract}
Infinite-state games are a commonly used model for the synthesis of reactive systems with unbounded data domains. 
Symbolic methods for solving such games need to be able to construct intricate arguments to establish the existence of winning strategies.
Often,  large problem instances require prohibitively complex arguments.
Therefore, techniques that identify smaller and simpler sub-problems and exploit the respective results for the given game-solving task are highly desirable.

In this paper, we propose the first such technique for infinite-state games.  The main idea is to enhance symbolic game-solving with the results of localized attractor computations performed in sub-games.  The crux of our approach lies in identifying useful sub-games by computing permissive winning strategy templates in finite abstractions of the infinite-state game.
The experimental evaluation of our method demonstrates that it outperforms existing techniques and is applicable to infinite-state games beyond the state of the art.
\end{abstract}

\section{Introduction}\label{sec:intro}
\input{intro}

\section{Preliminaries}\label{sec:prelim}
\input{preliminaries}

\section{Attractor Computation with Caching}\label{sec:using-cache}
\input{cache-utilization}

\section{Abstract Template-Based Cache Generation}\label{sec:computing-cache}
\input{overview-generation}

\subsection{Generating Attractor Caches from Sub-Games}\label{sec:subgame-cache}
\input{subgame-cache}

\subsection{Constructing Sub-Games from Abstract Strategy Templates}\label{sec:template-subgame}
\input{abstraction}

\input{template-subgame}

\section{Game Solving with Abstract Template-Based Caching}
\label{sec:full-solving}
\input{solving}
\input{pruning}

\section{Experimental Evaluation}\label{sec:evaluation}
\input{evaluation}

\section{Related Work}\label{sec:related}
\input{related}

\section{Conclusion}\label{sec:conclusion}
\input{conclusion}

\paragraph{Data Availability Statement}
The software generated during and/or analysed during the current study is available in the Zenodo repository~\cite{heim_2024_zenodo}.

\bibliographystyle{splncs04}
\bibliography{main.bib}

\newpage
\appendix

\section{Additional Definitions}\label{app:prelim}
\input{appendix-preliminaries}

\newpage

\section{Proofs}\label{app:proofs}
\input{appendix-proofs}
\newpage

\section{Pruning of the Winning Regions in the Abstract Game}\label{app:pruning}
\input{appendix-pruning}
\newpage

\section{Benchmarks}\label{app:benchmarks}
\input{appendix-benchmarks}

\end{document}

%% file: macros.tex
\newcommand{\changed}[1]{\textcolor{blue}{#1}}

\newcommand{\power}[1]{\ensuremath{2^{#1}}}
\newcommand{\sema}[1]{\ensuremath{\llbracket}#1\ensuremath{\rrbracket}}
\newcommand{\Nat}{\ensuremath{\mathbb{N}}}
\newcommand{\Bool}{\ensuremath{\mathbb{B}}}
\newcommand{\Int}{\ensuremath{\mathbb{Z}}}

\newcommand{\Real}{\ensuremath{\mathbb{R}}}

\newcommand{\dom}{\ensuremath{\mathit{dom}}}

\newcommand{\correspond}{\ensuremath{~\widehat{=}~}}

\newcommand{\vars}{\ensuremath{\mathit{Vars}}}
\newcommand{\funcSymbols}{\ensuremath{\Sigma_F}}

\newcommand{\funcTerms}{\ensuremath{\mathcal{T}_F}}

\newcommand{\values}{\ensuremath{\mathcal{V}}}
\newcommand{\functions}{\ensuremath{\mathcal{F}}}

\newcommand{\assignment}{\ensuremath{\nu}}
\newcommand{\assignments}[1]{\ensuremath{\mathit{Assignments}}(#1)}
\newcommand{\interpretation}{\ensuremath{\mathcal I}}
\newcommand{\interpretations}[1]{\ensuremath{\mathit{Interpretations}}(#1)}

\newcommand{\assmt}[1]{\mathbf{#1}}
\newcommand{\eval}[1]{\ensuremath{\chi_{#1}}}

\newcommand{\FOL}[1]{\ensuremath{\mathit{FOL}}(#1)}
\newcommand{\FOLX}{\ensuremath{\mathit{FOL}}}
\newcommand{\QF}[1]{\ensuremath{\mathit{QF}}(#1)}
\newcommand{\QFX}{\ensuremath{\mathit{QF}}}

\newcommand{\FOLentailsT}[1]{\ensuremath{\models_{#1}}}
\newcommand{\FOLentails}{\ensuremath{\models}}

\newcommand{\plays}{\mathit{Plays}}

\newcommand{\last}{\mathit{last}}
\newcommand{\str}[2]{\mathit{Strat}_{#1}(#2)}

\newcommand{\reach}{\mathit{Reach}}
\newcommand{\safe}{\mathit{Safe}}
\newcommand{\buchi}{\mathit{Buchi}}
\newcommand{\cobuchi}{\mathit{coBuchi}}
\newcommand{\parity}{\mathit{Parity}}
\newcommand{\col}{\mathit{col}}
\newcommand{\Omegafin}{\Omega_{fin}}
\newcommand{\Inf}{\mathit{Inf}}

\newcommand{\Inv}{\ensuremath{\mathit{Inv}}}
\newcommand{\cells}{\ensuremath{\mathbb{X}}}
\newcommand{\inputs}{\ensuremath{\mathbb{I}}}

\newcommand{\states}{\ensuremath{\mathcal{S}}}
\newcommand{\symstates}{\ensuremath{\mathcal{D}}}

\newcommand{\guards}{\ensuremath{\mathit{Guards}}}

\newcommand{\sys}{\ensuremath{\mathit{Sys}}}
\newcommand{\env}{\ensuremath{\mathit{Env}}}

\newcommand{\rpg}{\ensuremath{\mathcal{G}}}

\newcommand{\loc}{\ensuremath{\mathit{loc}}}
\newcommand{\Labels}{\ensuremath{\mathit{Labels}}}

\newcommand{\cpre}[2]{\ensuremath{\mathit{CPre}_{#1}(#2)}}

\newcommand{\rpgs}{\ensuremath{\mathsf{RPGS}}}

\newcommand{\predss}{\ensuremath{\mathcal P_{\cells}}}
\newcommand{\predsi}{\ensuremath{\mathcal P}_{\cells \cup \inputs}}

\newcommand{\concretize}{\gamma}
\newcommand{\abstracts}{\alpha}

\newcommand{\overapproxp}{\ensuremath{\mathit{OverapproxP}}}

\newcommand{\livegroup}{\mathcal{H}}

\newcommand{\livegroupSingleN}{H}
\newcommand{\colivegroup}{D}
\newcommand{\safegroup}{U}

\newcommand{\edgeso}{E_p}

\newcommand{\strat}{\sigma}
\newcommand{\play}{\rho}

\newcommand{\src}{\textsc{src}}

\newcommand{\prune}{\ensuremath{\mathtt{prune}}}
\newcommand{\sink}{\ensuremath{\mathsf{sink}}}

\newcommand{\subgame}{\ensuremath{\mathsf{SubGame}}}
\newcommand{\spaths}{\ensuremath{\mathsf{SimplePaths}}}

\newcommand{\extend}{\ensuremath{\mathsf{extend}}}

\newcommand{\independent}{\ensuremath{\mathit{ind}}}
\newcommand{\indvars}{\ensuremath{\mathsf{IndependentVars}}}

\newcommand{\helpful}{\ensuremath{\mathsf{Helpful}}}
\newcommand{\pre}{\ensuremath{\mathsf{Pre}}}
\newcommand{\post}{\ensuremath{\mathsf{Post}}}

\newcommand{\targ}{\ensuremath{\mathit{targ}}}
\newcommand{\sorc}{\ensuremath{\mathit{src}}}

\newcommand{\toolname}{\mbox{\textsf{rpg-STeLA}}\xspace}
\newcommand{\rpgsolve}{\mbox{\textsf{rpgsolve}}\xspace}
\newcommand{\pestel}{\mbox{\textsf{PeSTel}}\xspace}
\newcommand{\muval}{\mbox{\textsf{MuVal}}\xspace}

%% file: intro.tex
\emph{Games on graphs} provide an effective way to formalize the automatic synthesis of \emph{correct-by-design} software in cyber-physical systems. 
The prime examples are algorithms that synthesize \emph{control software} to ensure high-level logical specifications in response to external environmental behavior. 
These systems typically operate over unbounded data domains. 
For instance, in smart-home applications \cite{sylla2018modular}, they need to regulate real-valued quantities like room temperature and lighting in response to natural conditions, day-time, or energy costs. 
Also, unbounded data domains are valuable for over-approximating large countable numbers of products in a smart manufacturing line \cite{gueye2018discrete}. 
The tight integration of many specialized machines
makes their efficient control 
challenging. 
Similar control synthesis problems occur in robotic warehouse systems \cite{amazonRobotics}, underwater robots for oil-pipe inspections \cite{UnderwaterRobots}, and electric smart-grid regulation \cite{masselot2016towards}.

Algorithmically, the outlined synthesis problems can be formalized via \emph{infinite-state games} that model the ongoing interaction between the system (with its to-be-designed control software) and its environment over their \emph{infinite} data domains.
Due to their practical relevance and their challenging complexity, there has been an increasing interest in \emph{automated techniques} for solving infinite-state games to obtain correct-by-design control implementations. 
As the game-solving problem is in general undecidable in the presence of infinite data domains, this problem is substantially more challenging than its finite-state counterpart.

\input{figure-overview}

Within the literature\footnote{See \Cref{sec:related} for a detailed discussion of related work.}, there are two prominent directions to attack this problem.
One comprises \emph{abstraction-based approaches}, where either the overall synthesis problem (e.g.~\cite{HenzingerJM03,WalkerR14}) or the specification (e.g.\ \cite{FinkbeinerKPS19,ChoiFPS22,MaderbacherB22}) are abstracted, resulting in a finite-state game, to which classical techniques apply. 
The other one are \emph{constraint-based techniques}~\cite{FarzanK18,SamuelDK21, FaellaP23,SamuelDK23}, that work directly on a symbolic representation of the infinite-state game. 
Due to the undecidability of the overall synthesis problem, both categories are inherently constrained. 
While abstraction-based approaches are limited by the abstraction domain they employ, constraint-based techniques typically diverge due to non-terminating fixpoint computations. 

To address these limitations, a recent constraint-based technique called \emph{attractor acceleration}~\cite{HeimD24} employs \emph{ranking arguments} to improve the convergence of symbolic game-solving algorithms. While this technique has shown superior performance over the state-of-the art, 
the utilized ranking arguments become complex,  and thus difficult to find,  as the size of the games increases. 
This makes the approach from~\cite{HeimD24} infeasible in such cases, often resulting in divergence in larger and more complex games.

In this paper, we propose an approach to overcoming  the above limitation and thus extending the applicability of synthesis over infinite state games towards realistic applications.
The key idea is to utilize efficient abstraction-based pre-computations that  \emph{localize} attractor computations to \emph{small and useful sub-games}.
In that way, acceleration can be applied locally to small sub-games,  and the results utilized by the procedure for solving the global game.
This often avoids the computationally inefficient attractor acceleration over the complete game.
To \emph{guide} the identification of useful sub-games, our approach computes \emph{strategy templates} \cite{AnandNS23} -- a concise representation of a possibly infinite number of winning strategies -- in finite abstractions of the infinite-state game.
\Cref{fig:overview} shows an overview of our method which also serves as an outline of the paper.

Our experimental evaluation demonstrates the superior performance of our approach compared to the state of the art.
Existing tools fail on almost all benchmarks, while our implementation terminates within minutes.

To build up more intuition, 
we illustrate the main idea of our approach with the following example, which will also serve as our running example.

\input{example-running}

\begin{example}\label{ex:running}
    \Cref{fig:running} shows a reactive program game for a sample-collecting robot. The robot moves along tracks, and its position is determined by the integer program variable $\mathit{pos}$. 
    The robot remains in location $\mathit{base}$ until prompted by the environment  to collect $\mathit{inpReq}$ many samples. 
    It cannot return to $\mathit{base}$ until the required samples are collected, as enforced by the variable $\mathit{done}$. From the right position, it can enter the $\mathit{mine}$, where it must stay and collect samples from two sites, $a$ and $b$. However, it has to choose the correct site in each iteration, as they might not have samples all the time (if both do not have samples,  it can get one sample itself). Once enough samples are collected,  the robot can return to  $\mathit{base}$.
 The requirement on the robot's strategy is to be at \emph{base} infinitely often.

\emph{Attractor acceleration}~\cite{HeimD24} uses ranking arguments to establish that by iterating some strategy an unbounded number of times through some location, a player in the game can  enforce reaching a set of target states.
In this example, to reach $\mathit{samp} \geq \mathit{req}$ in location $\mathit{mine}$ (the target) the robot can iteratively increase the value of $\mathit{samp}$ by choosing the right updates (the iterated strategy).
This works, since if $\mathit{samp}$ is increased repeatedly, eventually $\mathit{samp} \geq \mathit{req}$ will hold (the ranking argument).
Establishing the existence of the iterated strategy (i.e. the robot can increment $\mathit{samp}$) is a game-solving problem,  since the behavior of the robot is influenced by the environment.
This game-solving problem potentially considers the whole game,  since the iterated strategy is not known a priori.
In addition, identifying locations where acceleration can be applied and finding the right ranking arguments is challenging.
This impacts the scalability and applicability of acceleration, making it infeasible for large games.

Consequently, our method aims to identify \emph{small and useful sub-games} and \emph{cache the results obtained by solving these sub-games}. 
In \Cref{ex:running}, a useful sub-game would be the game restricted to the $\mathit{mine}$ location with the target state $\mathit{samp} \geq \mathit{req}$.  Applying the acceleration technique to this sub-game,  provides the ranking argument  described earlier.
These cached results are then utilized to enhance the symbolic game-solving procedure for the entire game.

To identify these small and useful sub-games, we use \emph{permissive strategy templates}~\cite{AnandNS23} in finite-state abstracted games.
They describe a potentially infinite set of winning strategies using local conditions on the transitions of the game.  
These local conditions (in the abstract game) provide guidance about local behavior in the solution of the infinite-state game without solving it.
This local behavior (e.g. incrementing $\mathit{samp}$ in $\mathit{mine}$) induces our sub-games.

\end{example}

%% file: figure-overview.tex
\begin{wrapfigure}[19]{r}{0.44\textwidth}
\centering
\scalebox{0.67}{%
\begin{tikzpicture}[>=stealth']
 \draw[thick, line width=2pt,blue!30] (0.2,-0.2) rectangle (8,7);
    \node[anchor=south east,blue!50] at (8,-0.3) {\textsc{RPGCacheSolve} (Alg.~\ref{alg:rpg-enhanced}, Sec.~\ref{sec:full-solving})} ;

    \draw[thick, line width=2pt,blue!50] (0.4,0.3) rectangle (7.8,2);
    \node[anchor=south east,blue!80] at (7.8,0.3) {\textsc{RPGSolveWithCache}} ;

    \node[text width=5.5cm, text centered] at (3.6,1.2) {integrated within \textsc{RPGSolve} \cite{HeimD24} };
    \node[blue] at (3.6,1.6) {\textsc{AttractorAccCache} (Alg.~\ref{alg:use-cache}, Sec.~\ref{sec:using-cache})};

    \node[align=center,text width=5cm, text centered] (n1) at (3,7.5) {Reactive Program Game\\[-0.1cm] (Def.~\ref{def:reactive-program-games}, Sec.~\ref{sec:prelim})};
    \node[align=center,text width=5.5cm, text centered] (n2) at (3,5.6) {Finite-State Abstract Game\\[-0.1cm] (Def.~\ref{def:abstractgame}, Sec.~\ref{sec:template-subgame})};
    \node[align=center,text width=5cm, text centered] (n3) at (3,4.2) {Strategy Templates\\[-0.1cm] (Def.~\ref{def:winningStrategyTemplate}, Sec.~\ref{sec:prelim})};
    \node[align=center,text width=5cm, text centered] (n4) at (3,2.8) {Infinite-State Sub-Games\\[-0.1cm] (Def.~\ref{def:subgame}, Sec.~\ref{sec:subgame-cache})};

    \draw ($(n1.south)+(0,0.1)$) edge[-latex, line width=2pt] node[right,blue,yshift=-0.1cm] {\textsc{AbstractRPG} (Sec.~\ref{sec:template-subgame})} (n2);
    \draw (n2) edge[-latex, line width=2pt] node[right] {\textsc{SolveAbstract} \cite{AnandNS23}} (n3);
    \draw (n3) edge[-latex, line width=2pt] node[right,blue] {(Sec.~\ref{sec:template-subgame})} (n4);
    
    \draw (n4) edge[-latex, line width=2pt] node[right,blue,yshift=0.2cm] {\textsc{GenerateCache} (Sec.~\ref{sec:subgame-cache})} (3,1.8) ;
\draw[-latex,line width=2pt] ($(n1.south)-(0,0.3)$) -| ($(n3.south)-(2.5,0.3)$) -- ($(n3.south)-(0.2,0.3)$);
 \draw[-latex,line width=2pt] ($(n3.south)-(2.5,0.3)$) -- ($(n3.south)-(2.5,1.5)$) -- (2.5,1.9);      
\end{tikzpicture}}
\caption{Schematic paper outline; contributions highlighted in blue.}\label{fig:overview}
\end{wrapfigure}
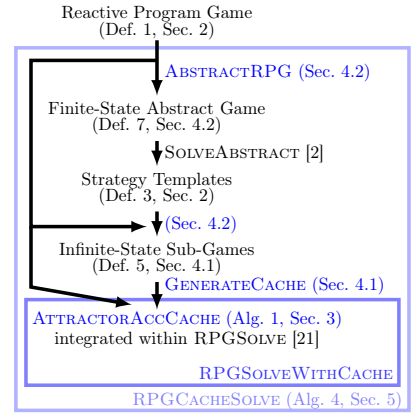

%% file: example-running.tex
\begin{figure}[b!]
\begin{center}
\caption{%
A reactive program game for a sample-collecting robot with locations ${\mathit{base}, \mathit{move}, \mathit{mine}}$, integer-type program variables ${\mathit{pos}, \mathit{done}, \mathit{req}, \mathit{samp}}$ and input variable ${\mathit{inpReq}}$. 
We use the following abbreviations: $\mathit{enterBase} \correspond (\mathit{pos} = 12 \land \mathit{done} = 1)$,
$\mathit{atMine} \correspond (\mathit{pos} = 23)$,
$\mathit{haveSamples} \correspond (a > 0 \lor b > 0)$,
$\mathit{enough} \correspond \mathit{samp} \geq \mathit{req}$,
$\mathit{sampleA} \correspond (\mathit{samp} := \mathit{samp} + a)$,
$\mathit{sampleB} \correspond (\mathit{samp} := \mathit{samp} + b)$, and
$\mathit{sampleS} \correspond (\mathit{samp} := \mathit{samp} + 1)$.
In each round of the game,  the environment chooses a value for the input ${\mathit{inpReq}}$. Based on guards over program variables and inputs, the game transitions to a black square. The system then chooses one of the possible updates to the program variables,  thus determining the next location. 
}\label{fig:running}
\vspace{3mm}
\begin{tikzpicture}[->,>=stealth',shorten >=1pt,auto,node distance=2.5cm]
  \tikzstyle{every state}=[fill=none,draw=black,text=black,inner sep=1.5pt, minimum size=16pt,thick,scale=0.6]

    \node[state] (mov) at (4,0) {$\mathit{move}$};
    \node[state,double] (bas) at (0,0) {$\mathit{base}$};
    \node[state] (col) at (8,0) {$\mathit{mine}$};

    \node[fill=black,draw=black,minimum size=0.5pt] (eBM) at (+2,-1) {};
    \node[fill=black,draw=black,minimum size=0.5pt] (eMB) at (1, 0) {};
    \draw[rounded corners] (bas) -- ++(0, 1) -- node[above] {\scriptsize $\mathit{inpReq} \leq 0$} ++ (1,0) -- (eMB);
    \draw[rounded corners] (bas) -- ++(0,-1) -- node[below] {\scriptsize $\mathit{inpReq} > 0$} (eBM);
    
    \draw[rounded corners]
       (eBM) -- node[below,align=left] {\scriptsize $\mathit{req} := \mathit{inpReq}$\\\scriptsize $\mathit{done} := 0$} ++ (1.8,0) -- (mov);

    \node[fill=black,draw=black,minimum size=0.5pt] (eCC0) at (11,0.5) {};
    \node[fill=black,draw=black,minimum size=0.5pt] (eCC1) at (11,-0.5) {};
    \node[fill=black,draw=black,minimum size=0.5pt] (eCM) at ( 6,-1) {};

    \draw[rounded corners] (col) -- ++(-0.2,-1) -- node[below] {\scriptsize $\mathit{enough}$} (eCM);
    \draw[rounded corners] (eCM) -- node[below] {\scriptsize $\mathit{done} := 1$} ++ (-1.8,0) -- (mov);

    \draw[rounded corners] (col) -- ++(0.2,-1) -- node[below, align=left] {\scriptsize $\lnot \mathit{enough} \land \lnot \mathit{haveSamples}$} ++ (2.8,0) -- (eCC1);
    \draw[rounded corners] (eCC1) -- node[above] {\scriptsize $\mathit{sampleS}$} ++(-2.6,0) -- (col);

    \draw[rounded corners] (col) -- ++(0,1) -- node[above, align=left] {\scriptsize $\lnot \mathit{enough} \land \mathit{haveSamples}$} ++ (3,0) -- (eCC0);
    \draw[rounded corners] (eCC0) -- node[above] {\scriptsize $\mathit{sampleA}$} ++(-2.6,+0.0) -- (col);
    \draw[rounded corners] (eCC0) -- ++(0,-0.5) -- node[above] {\scriptsize $\mathit{sampleB}$} (col);

    \node[fill=black,draw=black,minimum size=0.5pt] (eMC) at (7, 0) {};
    \node[fill=black,draw=black,minimum size=0.5pt] (eMM) at (4, 1) {};

    \draw (mov) -- node[above] {\scriptsize $\mathit{atMine}$} (eMC);
    \draw (mov) -- node[above] {\scriptsize $\mathit{enterBase}$} (eMB);

    \draw[rounded corners] (eMB) -- node[above] {} (bas);
    \draw[rounded corners] (eMC) -- node[above] {} (col);

    \draw (mov) -- node[left,align=right] {\scriptsize $\top$} (eMM);
    \draw[rounded corners] (eMM) -- node[above] {\scriptsize $\mathit{pos} := \mathit{pos} + 1$} ++(+2,0) -- ++(-1.6,-0.5) -- (mov);
    \draw[rounded corners] (eMM) -- node[above] {\scriptsize $\mathit{pos} := \mathit{pos} - 1$} ++(-2,0) -- ++(+1.6,-0.5) -- (mov);
\end{tikzpicture}
\end{center}
\end{figure}
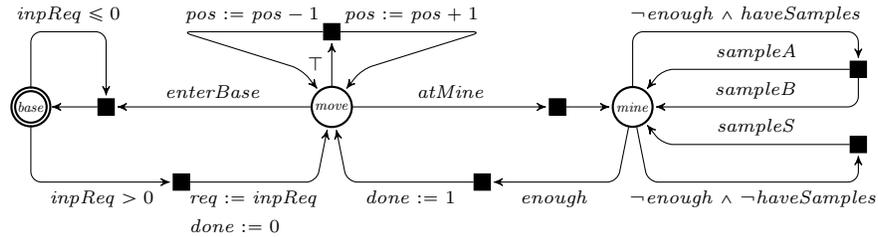

%% file: preliminaries.tex
\subsubsection{Sequences and First-Order Logic.}
For a set $V$, $V^*$ and $V^\omega$ denote the sets of finite, respectively infinite, sequences of elements of $V$, and let $V^\infty = V^* \cup V^\omega$.
For $\pi \in V^\infty$, we denote with $|\pi| \in \Nat \cup \{ \infty \}$ the length of $\pi$, and define $\dom(\pi):=\{0,\ldots, |\pi| - 1\}$. 
For $\pi = v_0v_1\ldots \in V^\infty$ and $i, j\in \dom(\pi)$ with $i<j$, we define $\pi[i]:=v_i$ and $\pi[i,j]:=v_i\ldots v_j$.
$\last(\pi)$ is the last element of a finite sequence $\pi$.\looseness=-1

Let $\values$ be the set of all values of arbitrary types, 
$\vars$ be the set of all variables,
$\functions$ be the set of all functions, and
$\funcSymbols$ be the set of all function symbols.
Let $\funcTerms$ be the set of all function terms defined by the grammar 
$\funcTerms \ni \tau_f ::= x \:|\: f(\tau_f^1, \dots \tau_f^n)$ for $f \in \funcSymbols$ and $x \in \vars$.
%
A function $\assignment: \vars \to \values$ is called an \emph{assignment}.
The set of all assignments over variables $X \subseteq \vars$ is denoted as $\assignments{X}$.
We denote the combination of two assignments $\assignment', \assignment''$ over disjoint sets of variables by $\assignment' \uplus \assignment''$.
%
A function $\interpretation: \funcSymbols  \to \functions$ is called an \emph{interpretation}.
The set of all interpretations is denoted as $\interpretations{\funcSymbols}$.
%
The evaluation of function terms $\eval{\assignment,\interpretation}: \funcTerms \to \values$ is defined by $\eval{\assignment,\interpretation}(x) := \assignment(x)$ for $x \in \vars$, $\eval{\assignment,\interpretation}(f(\tau_0, \dots \tau_n)) := \interpretation(f)(\eval{\assignment,\interpretation}(\tau_0), \dots \eval{\assignment,\interpretation}(\tau_n))$ for $f \in \funcSymbols$ and $\tau_0, \dots \tau_n \in \funcTerms$.
We denote the set of all first-order formulas as $\FOLX$ and by $\QFX$ the set of all quantifier-free formulas in $\FOLX$.
Let $\varphi$ be a formula and $X = \{x_1,\ldots,x_n\} \subseteq \vars$ be a set of variables.
We write $\varphi(X)$ to denote that the free variables of $\varphi$ are a subset of $X$.  
We also denote with $\FOL{X}$ and $\QF{X}$ the set of formulas (respectively quantifier-free formulas) whose free variables belong to $X$.
For a quantifier $Q \in \{\exists, \forall\}$, we write $Q X.\varphi$ as a shortcut for $Q x_1.\ldots Q x_n.\varphi$.  
We denote with $\FOLentails: \assignments{\vars} \times \interpretations{\funcSymbols} \times \FOLX$ the entailment of first-order logic formulas.
A \emph{first-order theory} $T \subseteq \interpretations{\funcSymbols}$ with $T \neq \emptyset$ restricts the possible interpretations of function and predicate symbols.
Given a theory $T$, for a formula $\varphi(X)$ and assignment $\assignment \in \assignments{X}$  we define that  
$\assignment \FOLentailsT{T} \varphi$ if and only if $\assignment, \interpretation \FOLentails \varphi$ for all 
$\interpretation \in T$.

For exposition on first-order logic and first-order theories, see c.f.~\cite{BradleyM07}.

\subsubsection{Two-Player Graph Games.}
A \emph{game graph} is a tuple $G = (V,V_\env,V_\sys,\rho)$ where $V = V_\env \uplus V_\sys$ are the vertices, partitioned between the environment player (\emph{player $\env$}) and the system player (\emph{player $\sys$}), and $\rho \subseteq (V_\env \times V_\sys) \cup (V_\sys \times V_\env)$ is the \emph{transition relation}.
A \emph{play} in $G$ is a sequence $\pi \in V^\infty$ where $(\pi[i],\pi[i+1])\in\rho$ for all $i \in \dom(\pi)$,  and if $\pi$ is finite then $\last(\pi)$ is a dead-end.\looseness=-1

For $p = \sys$ (or $\env$) we define $1 - p := \env$ (respectively $\sys$).
A \emph{strategy for player~$p$} is a partial function 
$\sigma: V^*V_{p} \to V$ where $\sigma(\pi\cdot v) = v'$ implies $(v,v') \in \rho$
and $\sigma$ is defined for all $\pi\cdot v \in V^*V_{p}$ unless $v$ is a dead-end.
$\str{p}{G}$ denotes the set of all strategies for player $p$ in $G$.
A play $\pi$ is \emph{consistent with $\sigma$ for player $p$} if $\pi[i+1] = \sigma(\pi[0,i])$ for every $i \in \dom(\pi)$ where $\pi[i] \in V_p$.
$\plays_G(v,\sigma)$ is the set of all plays in $G$ starting in $v$ and  consistent with strategy $\sigma$.
%

An \emph{objective} in $G$ is a set $\Omega \subseteq V^\infty$. A \emph{two-player turn-based game} is a pair $(G,\Omega)$, where $G$ is a game graph and $\Omega$ is an objective for player~$\sys$.
A sequence $\pi \in V^\infty$ is \emph{winning for player~$\sys$} if and only if $\pi \in \Omega$,  and is \emph{winning for player~$\env$} otherwise.  
We define different types of common objectives in \Cref{app:prelim}.
The \emph{winning region $W_p(G,\Omega)$ of player $p$ in $(G,\Omega)$} is the set of all vertices $v$ from which player $p$ has a strategy $\strat$ such that every play in $\plays_G(v,\sigma)$ is winning for player $p$. 
A strategy $\sigma$ of player $p$ is \emph{winning} if for every $v\in W_p(G,\Omega)$, every play in $\plays_G(v,\sigma)$ is winning for player~$p$.

\subsubsection{Acceleration-Based Solving of Infinite-State Games.}

We represent infinite-state games using the same formalism as~\cite{HeimD24}, called reactive program games. 
Intuitively,  reactive program games describe symbolically, using $\FOLX$ formulas and terms,  the possible interactions between the system player and the environment player in two-player games over infinite data domains.

\begin{definition}[Reactive Program Game Structure~\cite{HeimD24}]
\label{def:reactive-program-games}
A \emph{reactive program game structure} is a tuple 
$\mathcal G = (T,\inputs, \cells, L, \Inv,\delta)$ with the following components.
$T$ is a first-order theory.
$\inputs \subseteq \vars$ is a finite set of \emph{input variables}.
$\cells \subseteq \vars$ is a finite set of \emph{program variables} where 
$\inputs \cap \cells  = \emptyset$.         
$L$ is a finite set of \emph{game locations}.
$\Inv: L \to \FOL{\cells}$ maps each location to a \emph{location invariant}.
$\delta \subseteq L \times \QF{\cells \cup \inputs} \times (\cells \to \funcTerms) \times L$ is a finite \emph{symbolic transition relation} where

\begin{itemize}
    \item[(1)]  for every  $l \in L$ 
    the set of  \emph{outgoing transition guards}  $\guards(l) := \{ g \mid \exists u, l'.~(l, g, u , l') \in \delta \}$  is such that 
    $\bigvee_{g \in \guards(l)} g \equiv_{T} \top$,  and 
    for all $g_1, g_2 \in \guards(l)$ with $g_1 \neq g_2$ it holds that  $g_1 \land g_2 \equiv_{T} \bot$,
    \item[(2)] for all $l,g, u, l_1, l_2$,  if $(l, g, u, l_1) \in \delta$ and $(l, g, u, l_2) \in \delta$, then  $l_1 = l_2$, and
    \item[(3)] for every $l \in L$ and $\assmt{x} \in \assignments{\cells}$ such that $\assmt{x} \FOLentailsT{T} \Inv(l)$, and  $\assmt{i} \in \assignments{\inputs}$, there exist a transition $(l, g, u, l') \in \delta$ such that $\assmt{x} \uplus \assmt{i}\FOLentailsT{T} g$ and  $\assmt{x'} \FOLentailsT{T} \Inv(l')$ where $\assmt{x}'(x) =  \eval{\assmt{x}\uplus\assmt{i},\interpretation}(u(x))$ for all $x \in \cells$ and $\interpretation \in T$, and
    \item[(4)] for every $(l, g, u, l') \in \delta$, $f \in \funcSymbols(u)$, $\interpretation_1, \interpretation_2 \in T$ it holds that $\interpretation_1(f) = \interpretation_2(f)$.
\end{itemize}
\end{definition}
The requirements on  $\delta$ imply for each $l\in L$ that: 
(1) the guards in $\guards(l)$ partition the set $\assignments{\cells \cup \inputs}$,
(2) each pair of $g \in \guards(l)$ and update $u$  can label at most one outgoing transition from $l$,
(3) if there is an assignment satisfying the invariant at $l$, then for every input assignment  there is a possible transition, and
(4) the theory $T$ determines the meaning of functions in updates uniquely.
Given  locations $l,l' \in L$,  we define $\Labels(l,l'):=\{(g,u) \mid (l, g, u, l') \in \delta\}$ as the set of \emph{labels on transitions from $l $ to $l'$}.
We define as $\rpgs$ the set of all reactive program game structures.
The semantics of the reactive program game structure $\mathcal G$ is a (possibly infinite) game graph defined as follows.

\begin{definition}[Semantics of Reactive Program Game Structures]
\label{def:rpgs-semantics}
Let 
$\mathcal G = (T,\inputs, \cells, L, \Inv, \delta)$ be a reactive program game structure.
The semantics of $\mathcal G$ is the game graph 
$\sema{\mathcal G} = (\states, \states_\env,\states_\sys,\rho)$ where 
$\states := \states_\env\uplus\states_\sys$ and

    \begin{itemize}
        \item   $\states_\env := \{ (l,\assmt{x}) \in L \times \assignments{\cells} \mid \assmt{x} \FOLentailsT{T} \mathit{Inv}(l)\}$;
        \item   $\states_\sys :=\states_\env\times \assignments{\inputs}$;
        \item   $\rho \subseteq (\states_\env \times \states_\sys) \cup (\states_\sys \times \states_\env)$ is the smallest relation such that 
		\begin{itemize}
		\item $(s,(s,\assmt{i})) \in \rho$ for every $s \in \states_\env$ and $\assmt{i} \in \assignments{\inputs}$, 
		\item $(((l,\assmt{x}), \assmt{i}),(l',\assmt{x}')) \in \rho$ iff $\assmt{x}' \FOLentailsT{T} \mathit{Inv}(l')$ and there exists $(g,u) \in \Labels(l,l')$ such that $\assmt{x}\uplus\assmt{i} \FOLentailsT{T} g$,  $\assmt{x}'(x) =  \eval{\assmt{x}\uplus\assmt{i}, \interpretation}(u(x))$ for every $x \in \cells$ and $\interpretation \in T$.
        \end{itemize}
    \end{itemize}
\end{definition}

Note that this semantics differs from the original one in \cite{HeimD24} where the semantic game structure is not split into environment and system states.
We do that in order to consistently use the notion of a game graph.
Both semantics are equivalent.
We refer to the vertices of $\sema{G}$ as \emph{states}.
We define the function $\loc : \states \to L$ where $\loc(s) := l$ for any $s= (l,\assmt{x}) \in \states_\env$ and any $s = ((l,\assmt{x}),\assmt{i}) \in \states_\sys$.
By abusing notation,  we extend the function $\loc$ to sequences of states, defining $\loc : \states^\infty \to L^\infty$ where $\loc(\pi) = l_0l_1l_2\ldots$ iff $\loc(\pi[i]) = l_i$ for all $i \in \dom(\pi)$.
For simplicity of the notation, we write $W_p(\rpg,\Omega)$ instead of $W_p(\sema{\rpg},\Omega)$.
We represent and manipulate possibly infinite sets of states symbolically,  using formulas in \FOL{\cells} to describe sets of assignments to the variables in $\cells$.  
Our \emph{symbolic domain}  $\symstates := L \to \FOL{\cells}$ is the set of functions mapping locations to formulas in \FOL{\cells}. 
An element $d \in \symstates$ represents the states $\sema{d} : = \{((l,\assmt{x}) \in \states \mid  \assmt{x} \FOLentailsT{T}  d(l)\}.$
With $\{ l_1 \mapsto \varphi_1, \dots, l_n \mapsto \varphi_n \}$ we denote $d \in \symstates$ s.t.\ $d(l_i) = \varphi_i$ and $d(l) = \bot$ for $l \not\in \{l_1, \dots, l_n\}$.  
For brevity, we sometimes refer to elements of $\symstates$ as sets of states.

Note that the elements of the symbolic domain $\symstates$ represent subsets of $\states_\env$,  i.e.,  sets of environment states.  Environment states are pairs of location and valuation of the program variables. 
The system states, on the other hand,  correspond to intermediate configurations that additionally store the current input from the environment. This input is not stored further on (unless assigned to program variables). Thus,  we restrict the symbolic domain to environment states.

\paragraph{Solving Reactive Program Games.}
We consider \emph{objectives defined over the locations} of a reactive program game structure $\rpg$.
That is,  we require that if $\pi',\pi'' \in \states^\infty$ are such that $\loc(\pi') = \loc(\pi'')$, then $\pi' \in \Omega$ iff $\pi'' \in \Omega$. 
We consider the problem of solving reactive program games.
Given $\rpg$ and an objective $\Omega$ for Player~$\sys$ defined over the locations of $\rpg$, we want to compute $W_{\sys}(\sema{\mathcal G},\Omega)$.

\paragraph{Attractor Computation and Acceleration.}
A core building block of many algorithms for solving two-player games is the computation of \emph{attractors}.
Intuitively,  an attractor is the set of states from which a given player~$p$ can enforce reaching a given set of target states no matter what the other player does.
Formally, for a reactive program game structure $\rpg$, and $R \subseteq \states$ the 
\emph{player-$p$ attractor for $R$} is 
\centerline{$\mathit{Attr}_{\sema{\mathcal G},p}(R) := \{s \in \states \mid \exists \sigma \in \str{p}{\sema{\mathcal G}}.\forall \pi \in \plays_{\sema{\rpg}}(s,\sigma).\exists n\in \Nat.\;\pi[n]\in R\}.$}

In this work,  we are concerned with the symbolic computation of attractors in reactive program games. 
Attractors in reactive program games are computed using the so-called \emph{enforceable predecessor operator} over the symbolic domain $\symstates$.
For $d \in \symstates$, $\cpre{\rpg,p}{d} \in \symstates$ represents the states from which player $p$ can enforce reaching $\sema{d}$ in one step in $\rpg$ (i.e. one move by each player).  More precisely, \\
\centerline{$
\begin{array}{ll}
\sema{\cpre{\rpg,\sys}{d}} &= \{ s \in \states_\env \mid \forall s'. \; ((s,s') \in \rho) \rightarrow \exists s''.\; (s',s'') \in \rho \land  s ''\in \sema{d} \},~\text{and} \\
\sema{\cpre{\rpg,\env}{d}} &= \{s \in \states_\env \mid \exists s'. \; ((s,s') \in \rho) \land \forall s''.\; ((s',s'') \in \rho) \to s'' \in \sema{d}\}.\looseness=-1
\end{array}
$}
The player-$p$ attractor for $\sema{d}$ can be computed as a fixpoint of the enforceable predecessor operator: \\
\centerline{$
\begin{array}{ll}
\mathit{Attr}_{\sema{\mathcal G},p}(\sema{d}) \cap \states_\env &= \sema{\mu X. \;d \lor \cpre{\rpg,p}{X}},
\end{array}
$}
where $\mu$ denotes the least fixpoint.
Note that since $\states$ is infinite,  an iterative computation of the attractor is not guaranteed to terminate.

In~\Cref{ex:running},  consider the computation of the player-$\sys$ attractor for $\sema{d}$ where $d = \{ \mathit{base} \mapsto \top, \mathit{move} \mapsto \top, \mathit{mine} \mapsto \bot\}$.
Applying $\cpre{\rpg,\sys}{d}$ will produce $\{ \mathit{base} \mapsto \top, \mathit{move} \mapsto \top, \mathit{mine} \mapsto \mathit{samp} \geq \mathit{req}\}$ as in one step player-$\sys$ can enforce reaching $\mathit{move}$ if $\mathit{samp} \geq \mathit{req}$ in $\mathit{mine}$. 
Since in $\mathit{mine}$ the system player can enforce to increment $\mathit{samp}$ by at least one, a second iteration of $\cpre{\rpg,\sys}{\cdot}$ gives $\{\dots, \mathit{mine} \mapsto \mathit{samp} \geq \mathit{req} - 1\}$, a third  $\{\dots, \mathit{mine} \mapsto \mathit{samp} \geq \mathit{req} - 2\}$, and so on.
Thus,  a naive iterative fixpoint computation does no terminate here.
To avoid this non-termination, \cite{HeimD24} introduced \emph{attractor acceleration}. 
It will compute that, as explained in~\cref{sec:intro}, the fixpoint is indeed $\{\dots, \mathit{mine} \mapsto \top\}$.

\subsubsection{Permissive Strategy Templates.}

The main objective of this work is to identify small and useful sub-games,  for which the results can enhance the symbolic game-solving process. 
To achieve this, we use a technique called \emph{permissive strategy templates}~\cite{AnandNS23}, designed for finite game graphs. 
These templates can represent (potentially infinite) sets of winning strategies through local edge conditions. 
This motivates our construction of sub-games based on templates in \cref{sec:template-subgame}.
\looseness=-1


These strategy templates are structured using three local edge conditions: \emph{safety}, \emph{co-live}, and \emph{live-group} templates. 
Formally, given a game $(G,\Omega)$ with $G=(V,V_\env, V_\sys,\rho)$ and $\edgeso=\rho\cap(V_p\times V_{p-1})$, a \emph{strategy template for player $p$} is a tuple $(\safegroup,\colivegroup,\livegroup)$ consisting of a set of \emph{unsafe} edges $\safegroup\subseteq \edgeso$, a set of \emph{co-live} edges $\colivegroup\subseteq \edgeso$, and a set of live-groups $\livegroup \subseteq 2^{\edgeso}$. 
A strategy template $(\safegroup,\colivegroup,\livegroup)$ represents the set of plays $\Psi = \Psi_\safegroup \cap \Psi_\colivegroup \cap \Psi_\livegroup \subseteq \mathit{Plays}(G)$, where
\begin{align*}
    \Psi_\safegroup  &:=  \{ \pi \mid \forall i.~(\pi[i], \pi[i+1]) \not\in \safegroup \}, \ \ 
    \Psi_\colivegroup  :=  \{ \pi \mid \exists k.~\forall i > k.~(\pi[i], \pi[i+1]) \not\in \colivegroup \}, \\
    \Psi_\livegroup &:=  \bigcap_{\livegroupSingleN\in\livegroup} \{ \pi \mid (\forall i.~\exists j > i.~ \pi[j] \in \src(\livegroupSingleN)) \to (\forall i.~\exists j > i.~ (\pi[j], \pi[j+1]) \in \livegroupSingleN) \},
\end{align*}
where $\src(\livegroupSingleN)$ contains the sources $\{u \mid (u,v)\in \livegroupSingleN\}$ of the edges in $\livegroupSingleN$. 
A strategy $\strat$ for player $p$ \emph{satisfies} a strategy template $\Psi$ if it is winning in the game $(G,\Psi)$ for player $p$.
Intuitively, $\strat$ satisfies a strategy template if every play $\pi$ consistent with $\sigma$ for player $p$ is contained in $\Psi$, that is, (i) $\pi$ never uses the unsafe edges in~$\safegroup$ (i.e., $\pi\in\Psi_\safegroup$), (ii)  $\pi$ stops using the co-live edges in~$\colivegroup$ eventually (i.e., $\pi\in\Psi_\colivegroup$), and (iii) for every live-group $\livegroupSingleN\in\livegroup$, if $\play$ visits $\src(\livegroupSingleN)$ infinitely often, then it also uses the edges in $\livegroupSingleN$ infinitely often (i.e., $\pi\in\Psi_\livegroup$).
Strategy templates can be used as a concise representation of winning strategies as formalized next.

\begin{definition}[Winning Strategy Template~\cite{AnandNS23}]\label{def:winningStrategyTemplate}
    A strategy template $\Psi$ for player $p$ is \emph{winning} if every strategy satisfying $\Psi$ is winning for $p$ in $(G,\Omega)$.
\end{definition}
We note that the algorithms for computing winning strategy templates in safety, Büchi, co-Büchi, and parity games, presented in~\cite{AnandNS23}, exhibit the same worst-case computation time as standard methods for solving such (finite-state) games.

%% file: cache-utilization.tex
As outlined in~\cref{sec:intro},  the core of our method consists of the pre-computation of attractor sets for local sub-games and the utilization of the results in the attractor computations performed when solving  the complete reactive program game. 
We call the pre-computed results \emph{attractor cache}.
We use the cache during attractor computations to \emph{directly add} to the computed attractor sets of states from which, \emph{based on the pre-computed information},  the respective player can enforce reaching  the current attractor subset.
In that way,  if the local attractor computation requires acceleration,  we can avoid performing the acceleration during the attractor  computation for the overall game. 
This section presents the formal definition of an attractor cache and shows how it is used.

Intuitively, an attractor cache is a finite set of tuples or \emph{cache entries} of the form $(\rpg, p, \mathit{src}, \mathit{targ}, \cells_\independent)$. 
$\rpg$ is a reactive program game structure and $p$ the player the cache entry applies to.
The sets of states $\mathit{src}, \mathit{targ} \in \symstates$ are related via  enforceable reachability: player~$p$ can enforce reaching $\sema{\mathit{targ}}$ from $\sema{\mathit{src}}$ in $\rpg$. 
$\cells_\independent$ are the so-called \emph{independent variables} -- the enforcement relation must  hold independently of and preserve the values of $\cells_\independent$.
Independent variables are useful when a cache entry only concerns a part of the game structure where these variables are irrelevant. 
This allows the utilization of the cache entry under different conditions on those variables.
We formalize this intuition in the next definition.\looseness=-1

\begin{definition}[Attractor Cache]\label{defn:cache}
    A finite set $C \subseteq \rpgs \times \{\sys,\env\} \times \symstates \times \symstates \times \power{\cells}$ is called an \emph{attractor cache} if and only if for all $(\rpg, p, \mathit{src}, \mathit{targ}, \cells_\independent) \in C$ and all $\varphi \in \FOL{\cells_\independent}$ it holds that  
    $ \sema{\mathit{src} \land\lambda l. ~\varphi} \subseteq \mathit{Attr}_{\sema{\mathcal G},p}(\sema{\mathit{targ} \land \lambda l. ~\varphi}). $
\end{definition}

We use the \emph{lambda abstraction} $\lambda l. ~\varphi$ to denote the anonymous function that maps each location in $L$ to the formula $\varphi$.

\begin{example}\label{ex:cache}
Recall the game from~\Cref{ex:running}.
From every state with location $\mathit{mine}$,  player $\sys$ can enforce eventually reaching $\mathit{samp} \geq \mathit{req}$ by choosing at every step the update that increases variable $\mathit{samp}$.
As this argument only concerns location $\mathit{mine}$, the program variables $\mathit{done}$ and $\mathit{pos}$ are independent.
Since it is not updated,  $\mathit{req}$ is also independent (we prove this in the next section).
Hence, 
$C_\mathit{ex} = \{ (\rpg_\mathit{ex}, \mathit{Sys}, \mathit{src}, \mathit{targ}, \cells_\independent) \}$ where  $\rpg_\mathit{ex}$ is from~\Cref{fig:running},
$\mathit{src} = \{ \mathit{mine} \mapsto \top\}$, 
$\mathit{targ} = \{ \mathit{mine} \mapsto \mathit{samp} \geq \mathit{req}\}$, and $\cells_\independent = \{\mathit{done}, \mathit{pos}, \mathit{req} \}$ is an attractor cache.
\end{example}

\Cref{alg:use-cache} shows how we use an attractor cache to enhance accelerated attractor computations.
\textsc{AttractorAccCache} extends the procedure~\textsc{AttractorAcc} for accelerated symbolic attractor computation presented in~\cite{HeimD24}.
\textsc{AttractorAccCache} takes a cache  as an additional argument and at each iteration of the attractor computation checks if some cache entry is applicable.
For each such cache entry, if $\sema{\mathit{targ}}$ is a subset of $\sema{a^n}$, we can add $\mathit{src}$ to $a^n$ since we know that $\mathit{targ}$ is enforceable from $\mathit{src}$. 
However, $a^n$ may constrain the values of $\cells_\independent$ making this subset check fail unnecessarily. 
Therefore,  \textsc{StrengthenTarget} computes a formula $\varphi\in\FOL{\cells_\independent}$ such that $\mathit{targ}$ strengthened with $\varphi$ is a subset of $a^n$. 
Intuitively, $\varphi$ describes the values of the independent variables that remain unchanged in the cached attractor.
Note that $\varphi$ always exists as we could pick $\bot$, which we have to do if $\mathit{targ}$ is truly \emph{not} a subset of $a^n$.

The next lemma formalizes this intuition and the correctness of \textsc{AttractorAccCache} under the above condition on \textsc{StrengthenTarget}.
Note that since the cache is used in the context of attractor computation,  the objective $\Omega$ of the reactive program game is not relevant here. 

\SetAlgoSkip{0}
\SetKwComment{Comment}{/* }{ */}
\begin{algorithm}[t!]
\SetAlgoVlined
    \SetKwProg{Fn}{function}{}{}
    \DontPrintSemicolon
    \Fn{\textsc{AttractorAccCache}(
		 $\mathcal G$,  
    	 $\mathit{p} \in \{\sys,\env\}$,  
		 $d \in \symstates$, $C$: cache)}{
		\nl $a^0$ := $\lambda l.~\bot$;
        $a^1$ := $d$\;
        \nl \For{$n=1,2,\ldots$}{
			\nl \lIf{$a^{n} \equiv_{T} a^{n-1}$}{\Return $a^n$}
            \ForEach{$(\rpg', p', \mathit{src}, \mathit{targ}, \cells_\independent) \in C$ with $\rpg' = \rpg$ and $p'=p$ }{
            $\varphi : = \textsc{StrengthenTarget}(\mathit{targ},\cells_\independent,a^n)$\;\label{line:strengthen-cache}
            $a^n := a^n \;\lor\;(\mathit{src} \land (\lambda l.~\varphi))$\label{line:use-cache}
            }
			$a^n := a^n \;\lor\; \mathsf{Accelerate}(\mathcal{G}, p, l, a^n)$ \Comment{$ \mathsf{Accelerate}(...)$ is the result of applying attractor acceleration as in~\cite{HeimD24}}\label{line:accelerate}
			\nl $a^{n+1} := a^n \lor \cpre{\mathcal G,p}{a^n}$\label{line:cpre}\; 
			} 		
	}
\caption{Attractor computation using an attractor cache.}
\label{alg:use-cache}
\end{algorithm}

\begin{restatable}[Correctness of Cache Utilization]{lemma}{restateCacheUtilization}\label{lem:cache-utilization}
Let $\mathcal{G}$ be a reactive program game structure,  $p \in \{\sys,\env\}$, $d \in \symstates$ and $C$ be an attractor cache.
Furthermore, suppose that for every $\mathit{targ}\in\symstates$, $a \in \symstates$  and every $\cells_\independent \subseteq \cells$ it holds that if $\textsc{StrengthenTarget}(\mathit{targ}, \cells_\independent, a) = \varphi$,  then $\varphi \in \FOL{\cells_\independent}$ and $\sema{\mathit{targ} \land \lambda l. ~\varphi} \subseteq \sema{a}$.
Then,  if the procedure $\textsc{AttractorAccCache}(\mathcal G, p, d, C)$ terminates returning $\mathit{attr}\in \symstates$, then it holds that $\sema{\mathit{attr}} = \mathit{Attr}_{\sema{\mathcal G},p}(\sema{d}) \cap \states_\env$. 
\end{restatable}

We realize $\textsc{StrengthenTarget}(\mathit{targ},\cells_\independent,a)$ 
such that it returns the formula  $\bigwedge_{l \in L} \big(\forall (\cells \backslash \cells_\independent).~\mathit{targ}(l) \to  a(l)\big)$
which satisfies the condition in \Cref{lem:cache-utilization}.

\begin{example}
Recall the game from~\Cref{ex:running} and the cache $C_\mathit{ex}$ from~\Cref{ex:cache}.
Suppose that we are computing the attractor for player $\sys$ to $d = \{ \mathit{base} \mapsto \top \}$, i.e. $\textsc{AttractorAccCache}(\rpg_\mathit{ex}, \sys, d, C_\mathit{ex})$ without acceleration,  i.e.,  $\mathsf{Accelerate}$ returns $\bot$ in line~\ref{line:accelerate} in \Cref{alg:use-cache}.
Initially,  $a^1 = \{  \mathit{base} \mapsto \top\}$.
After one iteration of applying $\mathit{Cpre}$, we get $a^2 = \{\mathit{base} \mapsto \top, \mathit{move} \mapsto \mathit{pos} = 12 \land \mathit{done} = 1\}$. 
Then we get $a^3 = \{ \dots, \mathit{mine} \mapsto \mathit{pos} = 12 \land \mathit{samp} \geq \mathit{req}\}$. 
In the only entry of $C_\mathit{ex}$,  the target set $\mathit{targ} = \{ \mathit{mine} \mapsto \mathit{samp} \geq \mathit{req} \}$ contains more states in $\mathit{mine}$ (i.e., all possible positions of the robot) then $a^3$ (which asserts $\mathit{pos} = 12$). However, $\textsc{StrengthenTarget}(\mathit{targ}, \cells_\independent, a^3)$ as implemented above,  will return the strengthening $\mathit{pos} = 12$ (after simplifying the formula), which makes the cache entry  with $\mathit{targ}$ applicable.
Since $\mathit{src} = \{ \mathit{mine} \mapsto \top \}$, we  update $a^3$ to $\{ \dots, \mathit{mine} \mapsto \mathit{pos} = 12\}$ in line~\ref{line:use-cache} of the algorithm.
\end{example}

%% file: overview-generation.tex
\Cref{sec:using-cache} defined attractor caches and showed their utilization for attractor computations via \Cref{alg:use-cache}. 
We motivated this approach by the observation that there often exist \emph{small local sub-games} that entail essential attractors,  and pre-computing these attractors within the sub-games, caching them and then using them via \Cref{alg:use-cache} is more efficient then only applying acceleration over the entire game (as in \cite{HeimD24}). To formalize this workflow, \Cref{sec:subgame-cache} explains the generation of cache entries from sub-game structures of the given reactive program game,  and \Cref{sec:template-subgame} discusses the identification of helpful sub-game structures via permissive strategy templates in finite-state  abstractions of the given game.\looseness=-1

%% file: subgame-cache.tex
%
Within this subsection, we consider a sub-game structure $\rpg'$ which is induced by a subset of locations $L_\mathit{sub} \subseteq L$ of the original reactive program game structure $\rpg$, as formalized next. Intuitively, we remove all locations from $\rpg$ not in $L_\mathit{sub}$ and redirect their incoming transitions to a new sink location $\sink_\mathit{sub}$.

\begin{definition}[Induced Sub-Game Structure]\label{def:subgame}
   Let $\rpg = (T,\inputs, \cells, L, \Inv,\delta)$ be a reactive program game structure and let
   $L_\mathit{sub} \subseteq L$ be a set of locations.  
   The \emph{sub-game structure induced  by $L_\mathit{sub}$} is the reactive program game structure\\
   $\subgame(\rpg,L_\mathit{sub}) := (T,\inputs,\cells',L',\Inv',\delta')$ where 
   $L' := L_\mathit{sub} \cup \{\sink_{\mathit{sub}}\}$, \\
   $\cells' : = \{ x \in \cells \mid x~\text{appears in transitions from or invariants of } L_{\mathit{sub}} \text{ in } \rpg' \}$,\\
   $\Inv'(l) := \Inv(l)$ for all $l \in L' \setminus  \{\sink_{\mathit{sub}}\} $ and $\Inv'(\sink_{\mathit{sub}}) := \top$, and\newline
    $\begin{array}{lll}
   \delta' & := & \{(l,g,u,l') \in \delta \mid l,l' \in L' \} \cup 
                           \{(\sink_{\mathit{sub}},\top,\lambda x.\;x,\sink_{\mathit{sub}})\} \cup \\&&
                           \{(l,g,\lambda x.x,\sink_{\mathit{sub}}) \mid \exists l' \in L.\; (l,g,u,l') \in \delta \land l\in L' \land l' \not\in L'\}.
   \end{array}$
\end{definition}
Recall that $\symstates= L \to \FOL{\cells}$.  
Let $\symstates' := L' \to \FOL{\cells'}$ be the symbolic domain for a sub-game structure with locations $L'$.  
As $\cells' \subseteq \cells$, $\FOL{\cells'} \subseteq  \FOL{\cells}$ which allows us to extend each element of $\symstates'$ to an element of $\symstates$ that agrees on $L'$.
Formally,  we define $\extend_{L} : \symstates'  \to \symstates$ such that for $d'\in\symstates'$ and $l \in L$ we have 
$\extend_{L} (d')(l) := \texttt{if } l \in L' \texttt{ then }d'(l) \texttt{ else } \bot$.

\begin{algorithm}[b!]
\SetAlgoVlined
    \SetKwProg{Fn}{function}{}{}
    \DontPrintSemicolon
    \Fn{\textsc{SubgameCache}($\rpg$,  $p$, $L_\mathit{sub}$, $d \in \symstates$)}{
        $\rpg' = (T,\inputs,\cells',L',\Inv',\delta') := \subgame(\rpg,L_\mathit{sub})$\;\label{line:subgame-line}
        $d' := \lambda l.~\texttt{if } l \in L_\mathit{sub}\texttt{ then } \mathsf{QElim} (\exists (\cells \setminus \cells'). d(l)) \texttt{ else } \bot$\label{line:qelim-sgc}\;
        $a:= \textsc{AttractorAcc}(\rpg', p,d')$ \label{line:subgame-attractor}\Comment{attractor computation from~\cite{HeimD24}}
	    $\cells_\independent:=\indvars(\rpg,\rpg')$\;
        \Return $\{(\rpg,p,\extend_L(a),\extend_L(d'), \cells_\independent)\}$\;\label{line:subgame-output}
    }
\caption{Cache generation based on an induced sub-game.}
\label{alg:subgame-cache}
\end{algorithm}

\smallskip
The computation of an attractor cache from an induced sub-game is detailed in \Cref{alg:subgame-cache}. Given a reactive program game structure $\rpg$, a player $p$, and a subset of locations $L_\mathit{sub}$, \Cref{alg:subgame-cache} first computes the induced sub-game (line \ref{line:subgame-line}). 
The quantifier elimination (\cite[Ch.\ 7]{BradleyM07}) $\mathsf{QElim}$ in line~\ref{line:qelim-sgc}  projects the given $d \in \symstates$ to an element $d'$ of the symbolic domain $\symstates'$ of the sub-game structure.
Then, in line~\ref{line:subgame-attractor}, we perform the accelerated attractor computation from \cite{HeimD24} with target set $d'$   to obtain the set of states $a$ from which player $p$ can \emph{enforce} reaching $d'$ in $\rpg'$.
The independent variables are those variables in $\cells$ that are not updated in any of the transitions in $\rpg'$. Formally,  we define those as $\indvars(\rpg,\rpg'):=\{x \in \cells \mid \forall (l, g, u, l') \in \delta'.~u(x)=x\}$.
In order to output an attractor cache for the original game $\rpg$, we extend the computed source and target sets $a$ and $d'$ via the previously defined function $\extend_{L}$ (line \ref{line:subgame-output}). 
Intuitively, the attractor computed over a sub-game $\rpg'$ is also an attractor for the overall game $\rpg$ as  sub-games are only restricted by location (not by variables). Hence, player $p$ can also enforce reaching the target set in the original game $\rpg$, if he can do so in $\rpg'$. 
This is formalized by the next lemma. 

\begin{restatable}{lemma}{restateSubgameCache}\label{lemma:subgame-cache}
 Let $\rpg = (T,\inputs, \cells, L, \Inv,\delta)$ be a reactive program game structure, 
and let $\rpg'= (T,\inputs,\cells',L',\Inv',\delta')$ be an induced sub-game structure with sink location $\sink_{\mathit{sub}}$ constructed as above. 
Let $\sorc',\targ' \in \symstates'$ be such that $\targ'(\sink_{\mathit{sub}}) = \bot$ and $\sema{\sorc'} \subseteq \mathit{Attr}_{\sema{\mathcal G'},p}(\sema{\targ'})$ for some player $p\in\{\sys,\env\}$. 
Furthermore, let $\mathbb{Y} \subseteq \indvars(\rpg,\rpg')$.
Then,  for every $\varphi \in \FOL{\mathbb{Y} }$ it holds that\newline
\centerline{$\sema{\extend_L(\sorc') \land  \lambda l.\varphi } \subseteq 
\mathit{Attr}_{\sema{\mathcal G},p}(\sema{\extend_L(\targ') \land  \lambda l.\varphi} ).$}
\end{restatable}

This results in  the following correctness statement.

\begin{restatable}{lemma}{restateCorrectnessSubgameCache}\label{lem:correctness-subgame-cache}
   $\textsc{SubgameCache}(\rpg, p, L_\mathit{sub}, d)$ returns an attractor cache over $\rpg$.
\end{restatable}



\input{example-subgame}

\begin{example}\label{ex:subgame}
Consider the reactive program game structure $\rpg_\mathit{ex}$ from  \Cref{ex:running}. 
We apply $\textsc{SubgameCache}(\rpg_\mathit{ex}, \sys, \{\mathit{mine}\}, d)$ with $d =\{\mathit{mine} \mapsto \mathit{samp} \geq \mathit{req} \land \mathit{pos} = 12 \land \mathit{done} \neq 1 \}$.
First, we construct the induced sub-game structure in~\Cref{fig:subgame}.
Quantifier elimination  produces the target set $d' = \{ \mathit{mine} \mapsto \mathit{samp} \geq \mathit{req} \}$.
If we compute the attractor in this sub-game to set $d'$, we get $\{ \mathit{mine} \mapsto \top \}$.
Note that since the number of steps needed to reach $d'$ depends on the initial value of $\mathit{samp}$ and is hence unbounded,  a technique like acceleration~\cite{HeimD24} is necessary to compute this attractor.
As in this sub-game structure only the variable $\mathit{samp}$ is updated, the independent variables are $\cells_\independent = \{\mathit{done}, \mathit{pos}, \mathit{req} \}$.
With this we get the cache entry from~\Cref{ex:cache}.
\end{example}

%% file: example-subgame.tex
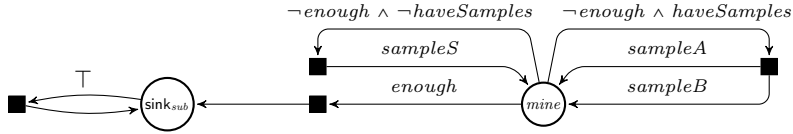
\begin{figure}[b!]
\begin{center}
\caption{Induced sub-game structure  $\subgame(\rpg_\mathit{ex},\{\mathit{mine}\})$ of  the reactive program game structure $\rpg_\mathit{ex}$ from \Cref{fig:running},  with the same abbreviations as in~\Cref{fig:running}.}\label{fig:subgame}
\begin{tikzpicture}[->,>=stealth',shorten >=1pt,auto,node distance=2.5cm]
  \tikzstyle{every state}=[fill=none,draw=black,text=black,inner sep=1.5pt, minimum size=16pt,thick,scale=0.65]

    \node[state] (mov) at (3,0) {$\sink_{\mathit{sub}}$};
    \node[state] (col) at (8,0) {$\mathit{mine}$};


    \node[fill=black,draw=black,minimum size=0.5pt] (sink) at (1,0) {};
    \draw (mov) edge[bend right = 10] node[above] {$\top$} (sink);
    \draw (sink) edge[bend right = 10] node {} (mov);

    \node[fill=black,draw=black,minimum size=0.5pt] (eCC0) at (11,0.5) {};
    \node[fill=black,draw=black,minimum size=0.5pt] (eCC1) at (5, 0.5) {};
    \node[fill=black,draw=black,minimum size=0.5pt] (eCM) at ( 5, 0) {};

    \draw[rounded corners] (col) -- node[above] {\scriptsize $\mathit{enough}$} (eCM);
    \draw[rounded corners] (eCM) -- node[above] {} (mov);

    \draw[rounded corners] (col) -- ++(-0.2,1) -- node[above, align=left,xshift=-5pt] {\scriptsize $\lnot \mathit{enough} \land \lnot\mathit{haveSamples}$} ++ (-2.8,0) -- (eCC1);
    \draw[rounded corners] (eCC1) -- node[above] {\scriptsize $\mathit{sampleS}$} ++(+2.6,+0.0) -- (col);

    \draw[rounded corners] (col) -- ++(0.2,1) -- node[above, align=left,xshift=5pt] {\scriptsize $\lnot \mathit{enough} \land \mathit{haveSamples}$} ++ (2.8,0) -- (eCC0);
    \draw[rounded corners] (eCC0) -- node[above] {\scriptsize $\mathit{sampleA}$} ++(-2.6,+0.0) -- (col);
    \draw[rounded corners] (eCC0) -- ++(0,-0.5) -- node[above] {\scriptsize $\mathit{sampleB}$} (col);

\end{tikzpicture}
\end{center}
\end{figure}

%% file: abstraction.tex
The procedure from the previous subsection yields attractor caches regardless of how the sub-games are chosen.
In this section we describe our approach to identifying ``useful'' sub-game structures.
These sub-game structures are induced by so-called \emph{helpful edges} determined by permissive strategy templates. 
Since the game graph described by a reactive program game structure is in general infinite, we first construct finite abstract games in which we compute permissive strategy templates for the two players. 
We start by describing the abstract games.

\subsubsection{Finite Abstractions of Reactive Program Games.}
Here we describe the construction of a game graph 
$\widehat G=(V, V_\env, V_\sys, \widehat\rho)$ from a reactive program game structure 
$\rpg = (T,\inputs, \cells, L, \Inv,\delta)$ with semantics 
$\sema{\rpg} = (\states, \states_\env,\states_\sys, \rho)$. 
While $\sema{\rpg}$ is also a game graph, its vertex set is typically infinite. 
The game graph $\widehat G$, which is an abstraction of $\sema{\rpg}$, has a finite vertex set instead. 

We construct the game graph $\widehat G$ from $\rpg$ by performing abstraction with respect to a given abstract domain. 
The abstract domain consists of two finite sets of quantifier-free first-order formulas which are used to define the vertex sets of the game graph $\widehat G$.
The conditions that we impose in the definition of abstraction domain given below ensure that it can partition the state space of $\rpg$.\looseness=-1

\begin{definition}[Game Abstraction Domain]
\label{def:abstraction-domain}
A \emph{game abstraction domain} for a reactive program game structure $\rpg =  (T,\inputs, \cells, L, \Inv,\delta)$ is a pair of finite sets of quantifier-free first-order formulas $(\predss,\predsi) \in \QF{\cells} \times \QF{\cells \cup \inputs}$ such that for $\mathcal{P} = \predss$ (resp.~$\mathcal{P}= \predsi$) and $V = \cells$ (resp.~$V = \cells \cup \inputs$), $\mathcal{P}$ partitions $\assignments{V}$, i.e. 
$\assignments{V} = \bigcup_{\varphi \in \mathcal{P}}\{ \assmt{v}  \mid \assmt{v} \FOLentailsT{T} \varphi \}$ 
and for every $\varphi_1,\varphi_2 \in \mathcal{P}$ with $\varphi_1\land \varphi_2$ satisfiable it holds that $\varphi_1 = \varphi_2$.
\end{definition}

The abstraction domain we use consists of all conjunctions of atomic predicates (and their negations) that appear in the guards of the reactive program game structure $\rpg$. 
Let $\mathit{GA}$ be the set of atomic formulas appearing in the guards of $\rpg$. 
We use the abstraction domain $\mathsf{AbstractDomain}(\rpg):=(\predss^{\mathit{GA}},\predsi^{\mathit{GA}})$ where 
\centerline{$
\begin{array}{ll}
\predss^{\mathit{GA}} & := \{ \bigwedge_{\varphi \in J}\varphi  \wedge 
\bigwedge_{\varphi \not\in J}  \neg\varphi \mid J \subseteq \mathit{GA} \cap \FOL{\cells}\}, \\
\predsi^{\mathit{GA}} & := \{ \bigwedge_{\varphi \in J}\varphi  \wedge 
\bigwedge_{\varphi \not\in J}  \neg\varphi \mid J \subseteq \mathit{GA} \cap (\FOL{\cells\cup\inputs}\setminus \FOL{\cells})\}. 
\end{array}
$
}

\begin{example}\label{ex:abstractiondomain}
In the game structure $\rpg_\mathit{ex}$ from \Cref{ex:running}, we get for 
$\predss^{\mathit{GA}}$ all combinations of $\varphi_1 \land \varphi_2 \land \varphi_3$, where $\varphi_1 \in \{\mathit{req} < \mathit{samp}, \mathit{req} \geq \mathit{samp} \}$, $\varphi_2 \in \{ \mathit{pos} = 12, \mathit{pos} = 23, \mathit{pos} \neq 12 \land \mathit{pos} \neq 23\}$, and $\varphi_3 \in \{ \mathit{task} = 1, \mathit{task} \neq 1 \}$.
For $\predsi^{\mathit{GA}}$ we get all combinations of $\psi_1 \land \psi_2 \land \psi_3$, where $\psi_1 \in \{a \leq 0, a > 0\}$, $\psi_2 \in \{ b \leq 0, b > 0\}$, and $\psi_3 \in \{\mathit{inpReq} \leq 0 , \mathit{inpReq} > 0 \}$.
\end{example}

We choose this abstraction domain as a baseline since the predicates appearing in the guards are natural delimiters in the program variable state space.
However, the abstraction we define now is independent of this specific domain.

Given a game abstraction domain $(\predss,\predsi)$,  we construct two abstract game graphs,   
$\widehat{G}^\uparrow$ and $\widehat{G}^\downarrow$.
They have the same sets of vertices but differ in the transition relations. 
The transition relation in $\widehat{G}^\uparrow$ overapproximates the transitions originating from states of player $\sys$  and underapproximates the transitions from states of player $\env$.  
In $\widehat{G}^\downarrow$ the approximation of the two players is reversed.

\begin{definition}[Abstract Game Graphs]\label{def:abstractgame}
Let $\rpg =  (T,\inputs, \cells, L, \Inv,\delta)$ be a reactive program game structure and 
 $(\predss,\predsi)$ be an abstraction domain. 
The game graphs $\widehat{G}^\circ=(V, V_\env, V_\sys, \widehat\rho^\circ)$ with $\circ\in\{\uparrow,\downarrow\}$ are the $(\predss,\predsi)$-induced abstractions of $\rpg$ if 
$V := V_\sys \cup V_\env$, $V_\env := L \times \predss$, and $V_\sys := L \times \predss \times \predsi  $; and 
$\widehat\rho^\circ \subseteq (V_\env \times  V_\sys) \cup (V_\sys  \times V_\env)$ is the smallest relation such that\\
\begin{inparaitem}
\item $((l, \varphi), (l, \varphi, \varphi_I)) \in \widehat\rho^\circ \cap (V_\env \times  V_\sys)$   iff the following formula is valid\\
\centerline{$\mathit{move}^\bullet_\cells(\Inv(l) \land \varphi(\cells), \exists \inputs.~ \varphi_I(\cells, \inputs))$}
\hspace{0.4cm}for $\bullet = \downarrow$ if $\circ = \uparrow$ and  $\bullet = \uparrow$ if $\circ = \downarrow$,\\
\item  $((l, \varphi, \varphi_I), (l', \varphi')) \in\widehat \rho^\circ \cap (V_\sys \times V_\env)$ iff the following formula is valid\\
\centerline{$\mathit{move}^\circ_{\cells \cup \inputs}\left(\Inv(l) \land \varphi(\cells) \land \varphi_I(\cells, \inputs), \exists (g , u) \in \mathit{Labels}(l, l').~ \mathit{trans}(g,u,l',\varphi')\right)$}
\hspace{0.4cm}for $\mathit{trans}(g,u,l',\varphi'):=g(\cells, \inputs) \land \big(\varphi'\land \Inv(l')\big)(u(\cells, \inputs))$,\\
\end{inparaitem}
where
$\mathit{move}^\uparrow_{V}(\varphi, \varphi')  :=  \exists V. \varphi(V) \land \varphi'(V)$ and
$\mathit{move}^\downarrow_{V}(\varphi, \varphi') :=  \forall V. \varphi(V) \to \varphi'(V)$.
\end{definition}

\Cref{def:abstractgame} provides us with a procedure \textsc{AbstractRPG} for constructing the pair of abstractions $(\widehat G^\uparrow,\widehat G^\downarrow) := \textsc{AbstractRPG}(\rpg,(\predss,\predsi))$.

We refer to the vertices in the abstract game graphs as abstract states.
By slightly overloading notation, we define the projection from abstract states $v\in V$ to the respective location by $\loc: V \to L$ s.t. $\loc((l,\varphi)) =l$ and $\loc((l,\varphi,\varphi_I)) =l$.
This definition naturally extends to sequences of abstract states 
$\pi\in V^\infty$ s.t.\ $\loc(\widehat \pi)[i] = \loc(\widehat \pi[i])$ for all $i \in dom(\pi)$,  and to sets of vertices: $\loc: 2^V \to 2^L$. 
%
Given  $\rpg $ with semantics $\sema{\rpg} = (\states, \states_\env,\states_\sys, \rho)$ and an abstract game graph $\widehat G^\circ=(V, V_\env, V_\sys, \widehat\rho^\circ)$, we define the following functions between their respective state spaces.
The \emph{concretization function} $\concretize: V \to 2^\states$  is defined s.t.\ \\
\centerline{$
\begin{array}{ll}
\concretize((l,\varphi)) &:= \{ (l, \assmt{x}) \in \states_\env \mid \assmt{x}  \FOLentailsT{T} \varphi \}~\text{and}\\
\concretize((l,\varphi,\varphi_I)) &:= \{ ((l, \assmt{x}),\assmt{i}) \in \states_\sys \mid \assmt{x} \uplus \assmt{i} \FOLentailsT{T} \varphi \land \varphi_I \}.
\end{array}
$}
The \emph{abstraction function} $\abstracts: \states \to 2^V$ is defined s.t.\ $v\in\abstracts(s)$ iff $s\in\concretize(v)$. 
We extend both function from states to (finite or infinite) state sequences 
$\pi\in \states^\infty$ and $\widehat\pi\in  V^\infty$ s.t. \\
\centerline{$
\begin{array}{ll}
\concretize(\widehat \pi) &:=\{ \pi \in \states^\infty  \mid |\pi|=|\widehat\pi| \land\forall i\in dom(\pi).~\pi[i] \in \concretize(\widehat \pi[ i])\},~\text{and}\\
\abstracts( \pi) &:=\{ \widehat\pi \in V^\infty  \mid |\pi|=|\widehat\pi| \land \forall i\in dom(\pi).~\widehat\pi[i] \in \abstracts(\pi[ i])\}.
\end{array}
$}
%
Both functions naturally extend to \emph{sets} of states or infinite sequences of states 
by letting 
$\concretize(A) : =  \bigcup_{a\in A} \concretize(a)$ 
for $A \subseteq V$ and $A\subseteq V^\omega$
and
$\abstracts(C) : =  \bigcup_{c\in C} \abstracts(c)$ 
for $C \subseteq \states$ and $C\subseteq \states^\omega$. 
Note that it follows from the partitioning conditions imposed on $(\predss,\predsi)$ in \Cref{def:abstraction-domain} that $\alpha$ is a total function and always maps states and state sequences to a singleton set. We abuse notation and write $\abstracts(s) = v$ (resp. $\abstracts(\pi) = \widehat\pi$) instead of $\abstracts(s) = \{v\}$ (resp. $\abstracts(\pi) = \{\widehat\pi\}$).

Let $\Omega \subseteq \states^\infty$ be an objective for the semantic game $\sema{\rpg}$.
With the relational functions $\langle\alpha,\gamma\rangle$ defined before,  $\Omega$ naturally induces an abstract objective $\widehat{\Omega} := \abstracts(\Omega)\subseteq V^\infty$ over the abstract state space $V$. 

Recall that we consider winning conditions $\Omega \subseteq \states^\infty$ for $\sema{\rpg}$ defined over the set $L$ of locations of $\rpg$.
As $\alpha$ preserves the location part of the states, $\widehat \pi \in \widehat\Omega$ iff $\concretize(\widehat \pi) \subseteq \Omega$.
That is, a sequence of abstract states is winning according to $\widehat \Omega$ iff all the corresponding concrete state sequences are winning according to $\Omega$.

The next lemma states the correctness property that the abstraction satisfies. 
More concretely,  $\widehat{G}^\uparrow$ overapproximates the winning region of player~$\sys$ in the concrete game, and $\widehat{G}^\downarrow$ underapproximates it.

\begin{restatable}[Correctness of the Abstraction]{lemma}{restateCorrectnessAbstraction}\label{lemma:abstraction}
Given a reactive program game structure $\rpg $ with semantics $\sema{\rpg} = (\states, \states_\env,\states_\sys, \rho)$ and location-based objective $\Omega$, let $\widehat G^\circ=(V, V_\env, V_\sys, \widehat\rho^\circ)$ with $\circ\in\{\uparrow, \downarrow\}$ be its $(\predss,\predsi)$-induced abstractions with relational functions $\langle\abstracts,\concretize\rangle$. Then it holds that 
(1) $W_{\sys}(\sema{\rpg},\Omega) \subseteq \concretize(W_{\sys}(\widehat{G}^\uparrow,\widehat\Omega))$, and
(2)   $\concretize(W_{\sys}(\widehat{G}^\downarrow,\widehat\Omega)) \subseteq W_{\sys}(\sema{\rpg},\Omega)$. 
\end{restatable}

By duality,  $\widehat{G}^\downarrow$
results in an overapproximation of the winning region of player~$\env$ in the concrete game.
Given an abstraction $\widehat{G}^\circ$, we denote with $\overapproxp(\widehat{G}^\circ)$ the player whose winning region is overapproximated in $\widehat{G}^\circ$:
$\overapproxp(\widehat{G}^\circ):= \sys$ if $\circ=\uparrow$ and  
$\overapproxp(\widehat{G}^\circ):= \env$ if $\circ=\downarrow$.

%% file: template-subgame.tex
\subsubsection{Abstract Strategy Templates and Their Induced Sub-Games.}

We now describe how we use a permissive strategy template for a player $p$ in an abstract game to identify sub-game structures of the given reactive program game from which to  generate attractor caches for player $p$.

We determine the sub-game structures and local target sets based on so-called \emph{helpful edges} for player $p$ in the abstract game where $p$ is over-approximated.
A helpful edge is a live-edge or an alternative choice to a co-live edge of a permissive strategy template.
Intuitively,  a helpful edge is an edge that player $p$ might have to take eventually in order to win the abstract game.
As our chosen abstraction domain is based on the guards,  a helpful edge often corresponds to the change  of conditions necessary to enable a guard in the reactive program game.
Since reaching this change might require an unbounded number of steps,   our method attempts a local attractor computation and potentially acceleration.
Identifying helpful edges based on permissive strategy templates rather than on winning strategies has the following advantages.
First,  templates reflect multiple abstract winning strategies for player $p$, capturing multiple possibilities to make progress towards the objective.
Moreover,  they describe local conditions,  facilitating the localization our method aims for.
Helpful edges are defined as follows.\looseness=-1

\begin{definition}[Helpful Edge]
Given a strategy template $(\safegroup,\colivegroup,\livegroup)$ for player $p$
in a game $(G,\Omega)$ with $G=(V,V_\env, V_\sys,\rho)$,   we call an edge $e \in \rho$ \emph{helpful for player $p$ w.r.t.~the template $(\safegroup,\colivegroup,\livegroup)$} if and only if the following holds:
There exists a live-group $H \in \livegroup$ such that $e \in H$, or
$e \not \in  \safegroup \cup \colivegroup$ and there exists a co-live edge $(v_s,v_t) \in \colivegroup$ with $v_s = \src(e)$.
We define $\helpful_{G,p}(\safegroup,\colivegroup,\livegroup)$ to be the set of helpful edges for player $p$ in $G$ w.r.t.~$(\safegroup,\colivegroup,\livegroup)$.
\end{definition}

For each helpful edge, we define \emph{pre-} and \emph{post-sets} which are the abstract environment states before and after that edge. This is formalized as follows.

\begin{definition}[Pre- and Post-Sets]
Let
$\widehat{G}^\circ=(V, V_\env, V_\sys, \widehat\rho^\circ)$ for some $\circ\in\{\uparrow,\downarrow\}$ be a 
$(\predss,\predsi)$-induced abstraction of  $\rpg$, 
let $p_{\mathit{over}}:= \overapproxp(\widehat{G}^\circ)$, and  
$e = (v_s,v_t) \in \helpful_{\widehat{G}^\circ,p_{\mathit{over}}}(\safegroup,\colivegroup,\livegroup)$ 
for some 
template $(\safegroup,\colivegroup,\livegroup)$.
If $p_{\mathit{over}} = \env$, we have $e \in V_\env \times V_\sys$ and define
$\pre(e,p_{\mathit{over}}) := \{ v_s \}$ and 
$\post(e,p_{\mathit{over}}) := \{ v \in V \mid (v_t,v) \in \widehat{\rho}^\circ \}$.
If $p_{\mathit{over}} = \sys$ we have that $e  \in V_\sys \times V_\env$ and define
$\pre(e,p_{\mathit{over}}) := \{ v \in V \mid (v,v_s) \in \widehat{\rho}^\circ \}$ and
$\post(e,p_{\mathit{over}}) := \{ v_t \}$.
Note that in  both cases it holds that $\pre(e,p_{\mathit{over}}), \post(e,p_{\mathit{over}}) \subseteq V_\env \subseteq L \times \predss$.
\end{definition}

\begin{algorithm}[t]
\SetAlgoVlined
    \SetKwProg{Fn}{function}{}{}
    \DontPrintSemicolon
    \Fn{\textsc{GenerateCache}($\rpg$,  $\widehat{G}^\circ$, 
    	 $p_{\mathit{over}}$,   $(\safegroup,\colivegroup,\livegroup)$,  $b \in \Nat$)}{
    	$\mathit{SubgameLocs} := \emptyset$,  $\mathit{PostSet} := \emptyset$,\;
    	\ForEach{$e \in \helpful_{\widehat{G}^\circ,p_{\mathit{over}}}(\safegroup,\colivegroup,\livegroup)$}{
			$L_S := \loc(\pre(e, p_{\mathit{over}}))$;
			$L_T := \loc(\post(e, p_{\mathit{over}}))$\;
            $L_\mathit{sub} := \{ l ~\mid~ \exists w \in \spaths(\rpg, L_S, L_T).~|w| \leq b \land \exists i. \;w[i] = l\}$\;
			$\mathit{SubgameLocs} := \mathit{SubgameLocs} \cup \{L_\mathit{sub}\}$\;
			$\mathit{PostSet}:=\mathit{PostSet} \cup \{(L_\mathit{sub},\post(e,p_{\mathit{over}}))\}$\;
		}
	   $C := \emptyset$\;
		\ForEach{$L_\mathit{sub} \in \mathit{SubgameLocs}$}{
			$\mathit{TargetSet} := \textsc{ConstructTargets}(L_\mathit{sub},\mathit{PostSet})$\Comment{see~\cref{eqn:constructtargets}}
			\ForEach{$\targ  \in \mathit{TargetSet}$}{			
		    	$C := C \cup \textsc{SubgameCache}(\rpg, p_{\mathit{over}}, L_\mathit{sub}, \targ)$\label{line:add-to-cache}\;
		    	}
		}
		\Return $C$
	}
\caption{Generation of a cache based on a strategy template.}
\label{alg:generate-cache}
\end{algorithm}

As a helpful edge represents potential ``progress'' for player $p$,  we consider the question of whether player $p$ has a strategy in the concrete game to reach the \emph{post-set} from the  \emph{pre-set}. 
This motivates the construction of sub-game structures induced by the locations connecting those two sets in the reactive program game.

Procedure \textsc{GenerateCache} in \Cref{alg:generate-cache} formalizes this idea. 
It takes an abstract game and a strategy template for the over-approximated player $p_\mathit{over}$ in this game.
For each helpful edge $e$,   it constructs  the sub-game structure induced by the set of locations that lie on a simple path in the location graph from the locations of the \emph{pre-set} to the \emph{post-set} of $e$.
The optional parameter $b$ allows for heuristically tuning the locality of the sub-games by bounding the paths' length.

For each sub-game structure, the target sets for the local attractor computations are determined by the post-sets of the helpful edges that induced this sub-game structure (it might be more than one). 
They are computed by 
\begin{equation}
\label{eqn:constructtargets}
\textsc{ConstructTargets}(L_\mathit{sub}, \mathit{PostSet}) = T_1 \cup T_2 \cup T_3
\end{equation}
where the sets $T_1,T_2 \text{ and }T_3$ of elements of $\symstates$ are defined as follows.
\begin{itemize}
\item $T_1:=\{d\in\symstates \mid \exists P. \;(L_\mathit{sub},P) \in \mathit{PostSet} \land \forall l \in L.\;d(l) = \bigvee_{(l,\varphi) \in P} \varphi\}$ consists of targets that are determined by a single post-set.
\item $T_2:=\{d_\cup\}$, where for every $l \in L$, $d_\cup(l) = \bigvee_{P \text{ s.t. } (L_\mathit{sub},P) \in \mathit{PostSet}} \bigvee_{(l,\varphi) \in P} \varphi$ is the singleton containing the union of the targets  of all post-sets.
\item $T_3:=\{d_\top\}$,  where for $l \in L$ ,  $d_\top(l)= \exists P, \varphi. (L_\mathit{sub},P) \in \mathit{PostSet} \land (\varphi, l) \in P$  contains the target that is $\top$ iff the location appears in some post-set.
\end{itemize}

Once the targets are  constructed, \textsc{GenerateCache} uses \textsc{SubgameCache} from \Cref{alg:subgame-cache} to compute the attractor caches for those targets and respective sub-game structures.
By \Cref{lem:correctness-subgame-cache}, \textsc{SubgameCache} returns attractor caches. 
As attractor caches are closed under set union,  we conclude the following.
\begin{corollary}\label{cor:cache-generation}
The set $C$ returned by  $\textsc{GenerateCache}$ is an attractor cache.
\end{corollary}

\begin{example}\label{ex:abstact-template}
    The abstractions of  $\rpg_\mathit{ex}$ from  \Cref{ex:running} and respective  templates are too large to depict.
    One helpful edge for $\sys$ is $e = ((\mathit{mine}, \varphi, \varphi_I), (\mathit{mine}, \varphi'))$ with
    $\varphi = \changed{\mathit{samp} < \mathit{req}} \land \mathit{pos} = 12 \land \mathit{done} \neq 1$, $\varphi' = \changed{\mathit{samp} \geq \mathit{req}} \land \mathit{pos} = 12 \land \mathit{done} \neq 1$, and $\varphi_I = a > 0 \land b \leq 0 \land \mathit{inpReq} \leq 0$.
    This edge $e$ is in a live group where the other edges are similar with different $\varphi_I$.
    They correspond to the situation where the value of $\mathit{samp}$ finally becomes greater or equal to $\mathit{req}$.
    For $e$,  $\pre(e, \sys) = \{(\mathit{mine}, \varphi)\}$ and $\post(e, \sys) = \{(\mathit{mine}, \varphi')\}$  result in $L_\mathit{sub} = \{\mathit{mine}\}$ and the  target $\{ \mathit{mine} \mapsto \varphi' \}$.
    With this, we generate a cache as in \Cref{ex:subgame}.\looseness=-1
\end{example}

%% file: solving.tex
\begin{algorithm}[t!]
\SetAlgoVlined
    \SetKwProg{Fn}{function}{}{}
    \DontPrintSemicolon
    \Fn{\textsc{RPGCacheSolve}(
		 $\mathcal G = (T,\inputs, \cells, L, \Inv,\delta)$,  
    	 $\Omega$, $b \in \Nat$)}{
	     $(\predss,\predsi):= \mathsf{AbstractDomain}(\rpg)$\;
	     $(\widehat G^\uparrow,\widehat G^\downarrow):=\textsc{AbstractRPG}(\rpg,(\predss,\predsi))$ \Comment{see ~\cref{def:abstractgame}}
         $C := \emptyset$\;
         \ForEach{$(p, \circ) \in \{(\sys,\uparrow), (\env, \downarrow)\}$}{
            $(\safegroup,\colivegroup,\livegroup) :=\textsc{SolveAbstract}(\widehat G^\circ,\Omega)$\;
            $C := C \cup \textsc{GenerateCache}(\rpg, \widehat{G}^\circ, p, (\safegroup,\colivegroup,\livegroup), b)$\;
         }
        \Return \textsc{RPGSolveWithCache}($\rpg, C$) \Comment{solves $\rpg$ using  \textsc{AttractorAccCache} in \cref{alg:use-cache} for attractor computation}
	}
\caption{Game solving with abstract template-based caching.}
\label{alg:rpg-enhanced}
\end{algorithm}

This section summarizes our approach for reactive progam game solving via \Cref{alg:rpg-enhanced}, which combines the procedures introduced in \Cref{sec:using-cache} and \Cref{sec:computing-cache} as schematically illustrated in \Cref{fig:overview} of \Cref{sec:intro}.
\Cref{alg:rpg-enhanced} starts by computing the abstract domain and both abstractions.
For each abstract game, $\textsc{SolveAbstract}$ computes a strategy template~\cite{AnandNS23}. Then,  \textsc{GenerateCache} is invoked to construct the respective attractor  cache.   
\textsc{RPGSolveWithCache} solves reactive program games in direct analogy to \textsc{RPGSolve} from \cite{HeimD24}, but instead of using \textsc{AttractorAcc}, it uses the new algorithm \textsc{AttractorAccCache} which utilizes the attractor cache $C$.
The overall correctness of \textsc{RPGCacheSolve} follows from \Cref{lem:cache-utilization}, \Cref{cor:cache-generation}, and the correctness of \cite{HeimD24}.

\begin{theorem}[Correctness] \label{thm:method-correctness}
Given a reactive program game structure $\rpg$ and a location-based objective $\Omega$,  for any $b \in \Nat$, if  \textsc{RPGCacheSolve} terminates, then it returns $W_{\sys}(\sema{\rpg},\Omega).$
\end{theorem}

%% file: pruning.tex
\begin{remark}
In addition to using the  strategy templates from the abstract games for caching,  we can  make use of the winning regions in the abstract games, which are computed together with the templates. 
Thanks to \Cref{lemma:abstraction}, we know that 
outside of its winning region in the abstract game the over-approximated player loses for sure.
Thus, we can \emph{prune} parts of the reactive program game that correspond to the abstract states where the over-approximated player loses. As our experiments show that the main performance advantage is gained by caching rather than pruning, we give the formal details for pruning in \Cref{app:pruning}.
%
%
\end{remark}

\paragraph{Discussion.}
The procedure \textsc{RPGCacheSolve} depends on the choice of game abstraction domain $(\predss,\predsi)$ and on the construction of the local games performed in \textsc{GenerateCache}.
The abstraction based on guards is natural,  as it is obtained from the predicates appearing in the game.  Acceleration~\cite{HeimD24} is often needed to establish that some guards can eventually be enabled.  Therefore,  we choose an abstraction domain that represents precisely the guards in the game.
 
Helpful edges capture transitions that a player might need to take, hence the game solving procedure has to establish that the player can eventually enable their guards. 
This might require acceleration, and hence motivates our use of helpful edges to construct the local games.
Investigating alternatives to these design choices and their further refinement is a subject of future work.

%% file: evaluation.tex
\input{evaluation-table.tex}

We implemented \Cref{alg:rpg-enhanced} for solving reactive program games in a prototype tool\footnote{Available at \url{https://doi.org/10.5281/zenodo.10939871}} \toolname (Strategy Template-based Localized Acceleration).
Our implementation is based on the open-source reactive program game solver \rpgsolve from \cite{HeimD24}. 
Specifically, we use \rpgsolve for the \textsc{AttractorAcc} and \textsc{RPGSolveWithCache} methods to compute attractors via acceleration and to solve reactive program games utilizing the precomputed cache, respectively. 
We realize \textsc{SolveAbstract} by using \pestel~\cite{AnandNS23}, which computes strategy templates in finite games.
We do not use the bound $b$ in \Cref{alg:rpg-enhanced}.

We compare our tool \toolname to the solver \rpgsolve and the $\mu$CLP solver \muval~\cite{UnnoSTK20}. 
Those are the only available techniques that can handle unbounded strategy loops, as stated in \cite{HeimD24}.
Other tools from \cite{NeiderT16,FarzanK18,MarkgrafHLNN20,BeyeneCPR14,FaellaP23,SamuelDK21,SamuelDK23,MaderbacherB22,ChoiFPS22} cannot handle those, are outperformed by \rpgsolve, or are not available.
For \muval, we encoded the games into $\mu$CLP as outlined in \cite{UnnoSTK20} and done in \cite{HeimD24}.

\paragraph*{Benchmarks.}

We performed the evaluation on three newly introduced sets of benchmarks (described in detail in \cref{app:benchmarks}).
They all have unbounded variable ranges, contain unbounded strategy loops, and have Büchi winning conditions.
On the literature benchmarks from~\cite{cinderella,BodlaenderHKSWZ12,NeiderT16,MaderbacherB22,HeimD24} \rpgsolve performs well as \cite{HeimD24} shows.
Hence, we did not use them as local attractor caches are unnecessary, and they are smaller than our new ones.
Our new benchmark categories are:\looseness=-1

\noindent\textit{(1) Complex Global Strategy (Scheduler and Item Processing).}
These benchmarks consist of a scheduler and an item processing unit. 
The core feature of these benchmarks is that the system needs to perform tasks that require complex global strategic decisions and local strategic decisions requiring acceleration.

\noindent\textit{(2) Parametric Benchmarks (Chains).}
These benchmarks each consist of two parametric chains of local sub-tasks requiring acceleration and local strategic reasoning and more lightweight global strategic reasoning.
The number of variables scales differently in both chains, showcasing differences in scalability.

\noindent\textit{(3) Simple Global Strategy (Robot and Smart Home).} 
These benchmarks represent different tasks for a robot and a smart home.
The robot moves along tracks (with one-dimensional discrete position) and must perform tasks like collecting several products.
The smart home must, e.g., maintain temperature levels and adjust blinds depending on whether the house is empty or on the current time of day.
These benchmarks need acceleration and local strategic reasoning, but their global reasoning is usually simpler and more deterministic.

\paragraph*{Analysis.}
The experimental results in Table~\ref{tab:all} demonstrate that local attractor pre-computation and caching have a significant impact on solving complex games.  This is evidenced by the performance of \toolname that is superior to the other two tools.
We further see that pruning (without caching) is not sufficient, which underscores the need to use more elaborate local strategic information in the form of an attractor cache. 
This necessitates the computation of strategy templates, and simply solving an abstract game is insufficient.
However, as pruning does not cause significant overhead, it offers an additional optimization.

%% file: evaluation-table.tex
\begin{table}[b!]
\centering
\caption{%
Evaluation Results. 
ST is the variable domain type (additional to $\Bool$).
$|L|$, $|\cells|$, $|\inputs|$ are the number of respective game elements.
We show the wall-clock running time in seconds for our prototype \toolname in three settings (one with normal caching, one with additional pruning, one that only prunes), \rpgsolve, and \muval (with clause exchange).
TO means timeout after 30 minutes, MO means out of memory (8GB).
We highlight in bold the fastest solving runtime result.
The evaluation was performed on a computer equipped with an Intel(R) Core(TM) i5-10600T CPU @ 2.40GHz.
}\label{tab:all}
\scalebox{0.73}{
\begin{tabular}[t]{|l|crcc||rrr|r|r|}
\hline
    \multirow{2}{*}{Name} & \multirow{2}{*}{ST} & \multirow{2}{*}{$|L|$} & \multirow{2}{*}{$|\cells|$} & \multirow{2}{*}{$|\inputs|$} & \multicolumn{3}{c|}{\toolname} & \multirow{2}{*}{~\rpgsolve} & \multirow{2}{*}{~\muval} \\
    & & & & & \multicolumn{1}{c}{normal} & ~pruning & ~~prune-only & & \\
\hline
\hline                                                                    
    scheduler                   &  \Int  &  7 & 3 & 3 &            110.3 &   73.43 &  202.23 &   99.57 & \textbf{52.66} \\
    item processing             &  \Int  &  7 & 4 & 2 &  \textbf{473.85} &  479.34 &      TO &      TO & TO \\
\hline
\hline                                                                    
    chain 4                     &  \Int  &  7 & 6 & 1 &  \textbf{128.02} &  128.48 &      TO &      TO & TO \\
    chain 5                     &  \Int  &  8 & 7 & 1 &  \textbf{410.90} &  413.75 &      TO &      TO & TO \\
    chain 6                     &  \Int  &  9 & 8 & 1 & \textbf{1464.86} & 1470.13 &      TO &      TO & TO \\
    chain 7                     &  \Int  & 10 & 9 & 1 &               TO &      TO &      TO &      TO & TO \\
\hline                                                                    
    chain simple 5              &  \Int  &  8 & 3 & 1 &   \textbf{27.54} &   29.10 & 1364.91 & 1362.38 & TO \\
    chain simple 10             &  \Int  & 13 & 3 & 1 &   \textbf{76.41} &   80.01 &      TO &      TO & TO \\
    chain simple 20             &  \Int  & 23 & 3 & 1 &  \textbf{236.74} &  244.53 &      TO &      TO & TO \\
    chain simple 30             &  \Int  & 33 & 3 & 1 &  \textbf{485.73} &  503.89 &      TO &      TO & TO \\
    chain simple 40             &  \Int  & 43 & 3 & 1 &  \textbf{813.05} &  826.67 &      TO &      TO & TO \\
    chain simple 50             &  \Int  & 53 & 3 & 1 & \textbf{1212.90} & 1240.36 &      TO &      TO & TO \\
    chain simple 60             &  \Int  & 63 & 3 & 1 & \textbf{1704.02} & 1718.39 &      TO &      TO & TO \\
    chain simple 70             &  \Int  & 73 & 3 & 1 &               TO &      TO &      TO &      TO & TO \\
\hline                                                                    
\hline                                                                    
    robot running (\Cref{ex:running}) 
                                &  \Int  &  3 & 4 & 3 &  \textbf{470.69} & 471.59 &      TO &        TO & TO \\

    robot repair                &  \Int  &  6 & 4 & 2 &          TO & 91.66 & \textbf{51.40} &      TO & TO \\
    robot analyze samples       &  \Int  &  6 & 3 & 1 &  \textbf{104.02} &  113.06 &  684.67 &  632.39 & TO \\
    robot collect samples v1    &  \Int  &  4 & 3 & 1 &   \textbf{22.89} &   26.89 &      TO &      TO & TO \\
    robot collect samples v2    &  \Int  &  3 & 4 & 1 &  \textbf{478.33} &  483.50 &      TO &      TO & TO \\
    robot collect samples v3    &  \Int  &  4 & 3 & 3 &   \textbf{60.55} &   65.76 &      TO &      TO & TO \\
    robot deliver products 1    &  \Int  &  6 & 5 & 1 &   \textbf{95.08} &  101.75 &      TO &      TO & TO \\
    robot deliver products 2    &  \Int  &  7 & 6 & 2 &  \textbf{724.20} &  741.01 &      TO &      TO & TO \\
    robot deliver products 3    &  \Int  &  7 & 6 & 3 & \textbf{1116.31} & 1133.57 &      TO &      TO & TO \\
    robot deliver products 4    &  \Int  &  7 & 6 & 4 & \textbf{1580.03} & 1615.72 &      TO &      TO & TO \\
    robot deliver products 5    &  \Int  &  7 & 6 & 5 &               TO &      TO &      TO &      TO & TO \\
\hline                                                                    
    smart home day not empty    & \Real  &  5 & 5 & 2 &   \textbf{84.17} &  100.99 &      TO &      TO & TO \\
    smart home day warm         & \Real  &  6 & 5 & 3 &  \textbf{162.06} &  187.82 &      TO &      TO & TO \\
    smart home day cold         & \Real  &  6 & 5 & 3 &  \textbf{162.06} &  193.81 &      TO &      TO & TO \\
    smart home day warm or cold & \Real  &  6 & 5 & 4 &  \textbf{320.27} &  380.88 &      TO &      TO & TO \\
    smart home day empty        & \Real  &  5 & 5 & 2 &               TO &      TO &      TO &      TO & MO \\
    smart home night sleeping   & \Real  &  6 & 5 & 2 &   \textbf{80.69} &   99.38 &      TO &      TO & MO \\
    smart home night empty      & \Real  &  6 & 5 & 2 &               TO &      TO &      TO &      TO & TO \\
    smart home nightmode        & \Real  &  6 & 6 & 3 &               TO &      TO &      TO &      TO & MO \\
\hline
\end{tabular}}
\end{table}

%% file: related.tex
A body of methods for solving infinite-state games and synthesizing reactive systems operating over unbounded data domains exists.
Abstraction-based approaches reduce the synthesis problem to the finite-state case. 
Those include abstraction of two-player games~\cite{WalkerR14,VechevYY13, GrumbergLLS07,HenzingerJM03,FinkbeinerMPSS22}, which extends ideas from verification, such as abstract interpretation and counterexample-guided abstraction refinement, to games. 
The temporal logic LTL has recently been extended with data properties,  resulting in TSL~\cite{FinkbeinerKPS19} and its extension with logical theories~\cite{FinkbeinerHP22}. 
Synthesis techniques for those~\cite{FinkbeinerKPS19,ChoiFPS22,MaderbacherB22} are based on propositional abstraction of the temporal specification and iterative refinement by introducing assumptions.
The synthesis task's main burden in abstraction-based methods falls on the finite-state synthesis procedure.
In contrast,  we use abstraction not as the core solving mechanism but as a means to derive helpful sub-games.
Another class of techniques reason directly over the infinite-state space.  Several constraint-based approaches~\cite{FaellaP23, FarzanK18,KatisFGGBGW18} have been proposed for specific types of objectives.
\cite{SamuelDK21, SamuelDK23} lift fixpoint-based methods for finite-state game solving to a symbolic representation of infinite state sets.
However, a naive iterative fixpoint computation can be successful on a relatively limited class of games.  
Recently, \cite{HeimD24} proposed a technique that addresses this limitation by accelerating symbolic attractor computations.
However, as we demonstrate, their approach has limited scalability when the size of the game structure grows.
Our method mitigates this by identifying small helpful sub-games and composing their solutions to solve the game.

There are many approaches for compositional synthesis from LTL specifications~\cite{FiliotJR11,FinkbeinerGP23,FinkbeinerP20}. 
To the best of our knowledge, no techniques for decomposing infinite-state games exist prior to our work.
\looseness=-1

In verification, acceleration \cite{FinkelL02,  BardinFLP03,BardinFLS05} and  loop summarization~\cite{KroeningSTTW13} are applied to the loops in \emph{given program} and can thus be easily combined with subsequent analysis.
In contrast, in the setting of games, acceleration relies on \emph{establishing the existence of a strategy} which needs more guidance.

Permissive strategy templates were introduced in~\cite{AnandMNS23} and used in~\cite{AnandNS23} to represent sets of winning strategies for the system player in two-player games. 
They were used to synthesize hybrid controllers for non-linear dynamical systems~\cite{NayakEDS23}. Similar to our work, \cite{NayakEDS23} uses templates over abstractions to localize the compuation of continuous feedback controllers. While this inspired the solution methodology for infinite-state systems developed in this paper, the abstraction methodology and the semantics of the underlying system and its controllers are very different in \cite{NayakEDS23}. Our work is the first which uses permissive strategy templates as a guide for localizing the computation of fixpoints in infinite-state games.

%% file: conclusion.tex
We presented a method that extends the applicability of synthesis over infinite-state games towards realistic applications. 
The key idea is to reduce the  game solving problem to smaller and simpler sub-problems by utilizing winning strategy templates computed in finite abstractions of the infinite-state game.
The resulting sub-problems are solved using a symbolic method based on attractor acceleration.
Thus,  in our approach abstraction and symbolic game solving work in concert, using strategy templates as the interface between them. 
This opens up multiple avenues for future work,  such as exploring different abstraction techniques, as well as developing data-flow analysis techniques for reactive program games that can be employed in the context of symbolic game-solving procedures.

%% file: appendix-preliminaries.tex
\subsection{Concrete Types of Objectives}

A \emph{reachability objective}, denoted by $\reach(R)$, is defined via a set of vertices $R \subseteq V$ that player~$\sys$ is required to reach unless it ends in a dead-end vertex that belongs to the environment.  Formally, 
$\reach(R):= \{\pi\in V^\infty \mid \exists i\in \dom(\pi).\; \pi[i]  \in R \} \cup \{\pi\in V^* \mid \last(\pi) \in V_\env\}$.

A \emph{safety objective}, denoted by $\safe(S)$, is defined via a set of safe vertices $S \subseteq V$ that player~$\sys$ is required to stay within (until it ends in an environment dead-end if the play is finite).  Formally, 
$\safe(S) := \{\pi\in V^\omega \mid \forall i\in \Nat.\; \pi[i] \in S \} \cup  \{\pi\in V^* \mid \forall i\in \dom(\pi).\; \pi[i] \in S, \last(\pi) \in V_\env \}$.

Except for reachability and safety, all other objectives we consider are \emph{prefix-independent}.
Formally, an objective $\Omega \subseteq V^\infty$ is prefix-independent if for all $\pi \in V^\omega$ and $\tau \in V^*$ it holds that $\pi \in \Omega$ if and only if $\tau \cdot \pi \in \Omega$.
For such objectives, player~$\sys$ is required to ensure certain liveness properties unless the play ends in an environment dead-end.
Consequently, any (finite) sequence in $\Omegafin = \{\pi\in V^*\mid \last(\pi) \in V_\env\}$ is considered winning for player~$\sys$ irrespective of the liveness part.
Furthermore, the liveness part is defined using the set of vertices infinitely often in the sequences, i.e., 
$\Inf(\pi) = \{v \mid \forall i\in \Nat.\exists j > i.\; \pi[j] = v\}$.

A \emph{Büchi objective} $\buchi(B)$ 
for a set of accepting vertices $B \subseteq V$,  
requires that $B$ is visited infinitely often:
$\buchi(B) := \{\pi\in V^\omega \mid B \cap \Inf(\pi) \neq \emptyset\}\cup \Omegafin$. 

Its dual, \emph{co-Büchi objective} $\cobuchi(C)$ for rejecting vertices $C \subseteq V$ requires that $C$ is visited only finitely many times:
$\cobuchi(C) := \{\pi\in V^\omega \mid C \cap \Inf(\pi) = \emptyset \} \cup \Omegafin$.

A \emph{parity objective}, denoted by $\parity(\col)$, is defined via a function $\col : V \to \{0, 1, \dots k \}$ that associates for each vertex with a \emph{color} from $\{0, 1, \dots k\}$ .
It requires that the maximal color a play visits infinitely often is even. 
Formally, 
$\parity(\col) := \{\pi \in V^\omega \mid \max \{\col(v) \mid v\in \Inf(\pi)\} \text{ is even}\}\cup\Omegafin$.

In this work we consider \emph{location-based objectives} for reactive-program games. 
For instance,  we consider safety, reachability,  Büchi,  co-Büchi objectives defined via a set of locations in the reactive program game.

%% file: appendix-proofs.tex
\restateCacheUtilization*
\begin{proof}
We prove by induction that $\sema{a^n} \subseteq \mathit{Attr}_{\sema{\mathcal G},p}(\sema{d}) \cap \states_\env$ at every iteration $n$.  
Since $a^0 = \lambda l.~\bot$ and $a^1 = d$,  the statement holds for $n=0$ and $n=1$.

Suppose that the statement is satisfied for some $n \geq 1$. 
To show that it holds for $n+1$,  we will establish that all the (symbolically represented) sets of states  added to $a^n$ in order to obtain $a^{n+1}$ are subsets of $\mathit{Attr}_{\sema{\mathcal G},p}(\sema{d}) \cap \states_\env$.

For $\mathsf{Accelerate}(\mathcal{G}, p, l, a^n)$ in  line~\ref{line:accelerate} in \Cref{alg:use-cache},  the desired property is implied by the soundness of attractor acceleration established in~\cite{HeimD24} and the induction hypothesis.  
For $\cpre{\mathcal G,p}{a^n}$ in line~\ref{line:cpre}, the property follows from the definition of $\cpre{\cdot}{\cdot}$ and the induction hypothesis.
Thus,  it remains to show that for each $\sorc \land \lambda l. \;\varphi$ at line~\ref{line:use-cache}  it holds that 
$\sema{\sorc \land \lambda l. \;\varphi} \subseteq \mathit{Attr}_{\sema{\mathcal G},p}(\sema{d}) \cap \states_\env$.

According to line~\ref{line:strengthen-cache} in $\textsc{AttractorAccCache}$ and the assumption on the function $\textsc{StrengthenTarget}$,  we have that $\varphi \in \FOL{\cells_\independent}$ and 
$\sema{\mathit{targ} \land \lambda l. ~\varphi} \subseteq \sema{a^n}$.  
By induction hypothesis, 
$\sema{a^n}\subseteq \mathit{Attr}_{\sema{\mathcal G},p}(\sema{d})$.
Thus, we can conclude that  
$\mathit{Attr}_{\sema{\mathcal G},p}(\sema{\mathit{targ} \land \lambda l. ~\varphi}) \subseteq \mathit{Attr}_{\sema{\mathcal G},p}(\sema{d})$. 
Since $(\rpg, p, \mathit{src}, \mathit{targ}, \cells_\independent) \in C$, 
$C$ is an attractor cache,  and  $\varphi \in \FOL{\cells_\independent}$,
by \Cref{defn:cache} we have that $\sema{\sorc \land \lambda l. \;\varphi} \subseteq 
\mathit{Attr}_{\sema{\mathcal G},p}(\sema{\mathit{targ} \land \lambda l. ~\varphi})$.
Together with the previous inclusion, this implies that 
$\sema{\sorc \land \lambda l. \;\varphi} \subseteq \mathit{Attr}_{\sema{\mathcal G},p}(\sema{d})$.
As $\sema{\sorc \land \lambda l. \;\varphi} \subseteq \states_\env$ by the definition of $\sema{\cdot}$,  we can conclude that $\sema{\sorc \land \lambda l. \;\varphi} \subseteq \mathit{Attr}_{\sema{\mathcal G},p}(\sema{d}) \cap \states_\env$.

This completes the proof by induction and establishes the claim of the lemma.\qed
\end{proof}

\restateSubgameCache*
\begin{proof}
To prove the statement of the lemma, we define a reactive program game structure $\rpg''$ obtained from $\rpg'$ by extending it with the variables $\cells\setminus\cells'$,  that is, the variables in $\rpg$ that do not appear in $\rpg'$.  Formally,  we define
$$\rpg'' := (T,\inputs,\cells,L',\Inv',\delta''),$$ where 
$(l,g,u,l') \in \delta''$ if and only if there exists $(l,g,u',l') \in \delta'$  such that
\begin{itemize}
\item $u(x) = u'(x)$ for every $x \in \cells'$ and
\item $u(x) = x$ for every $x \in \cells\setminus\cells'$.
\end{itemize}

Note that by the definition of $\rpg''$ we have that 
$\indvars(\rpg,\rpg'') = \indvars(\rpg,\rpg')$.

Recall that $\symstates= L \to \FOL{\cells}$ and $\symstates' = L' \to \FOL{\cells'}$ are the symbolic domains associated with $\rpg$ and $\rpg'$ respectively. 
The symbolic domain for $\rpg''$ is $\symstates'' := L' \to \FOL{\cells}$.  
Clearly, since $\FOL{\cells'} \subseteq \FOL{\cells}$, we have that $\symstates' \subseteq \symstates''$.
We lift the function $\extend_L$ to $\symstates''$ defining it in the same way as for $\symstates'$.

We will use $\rpg''$ as an intermediate reactive program game structure in order to establish the desired relationship between $\rpg'$ and $\rpg$.
To this end, we establish the following properties that relate the three game structures.

\paragraph{Property 1.} 
For every $\sorc,\targ \in \symstates'$ we have that  if 
$\sema{\sorc} \subseteq \mathit{Attr}_{\sema{\rpg'},p}(\sema{\targ})$,  
then it also holds that 
$\sema{\sorc} \subseteq \mathit{Attr}_{\sema{\rpg''},p}(\sema{\targ})$.

\paragraph{Property 2.}
Let $\mathbb{Y} \subseteq \indvars(\rpg,\rpg'')$ and $\varphi \in \FOL{\mathbb{Y}}$.
Then,  for every $d \in \symstates''$ it holds that 
$$\mathit{Attr}_{\sema{\mathcal G''},p}(\sema{d}) \cap \sema{\lambda l.\varphi}\subseteq 
  \mathit{Attr}_{\sema{\mathcal G''},p}(\sema{d \land \lambda l.\varphi}).$$

\paragraph{Property 3.}
For every $\sorc,\targ \in \symstates''$  with $\targ(\sink_{\mathit{sub}}) = \bot$ we have that  if 
$\sema{\sorc} \subseteq \mathit{Attr}_{\sema{\rpg''},p}(\sema{\targ})$,  
then it also holds that 
$$\sema{\extend_L(\sorc)} \subseteq \mathit{Attr}_{\sema{\rpg},p}(\sema{\extend_L(\targ)}).$$

Before we prove these three properties, we will show that together they imply the statement of the lemma.
Since $\sema{\sorc'} \subseteq \mathit{Attr}_{\sema{\mathcal G'},p}(\sema{\targ'})$, 
\emph{Property 1} entails
$$\sema{\sorc'} \subseteq \mathit{Attr}_{\sema{\mathcal G''},p}(\sema{\targ'}).$$
Since  $\mathbb{Y} \subseteq \indvars(\rpg,\rpg')= \indvars(\rpg,\rpg'')$,  \emph{Property 2} yields
$$\mathit{Attr}_{\sema{\mathcal G''},p}(\sema{\targ'}) \cap \sema{\lambda l.\varphi}\subseteq 
  \mathit{Attr}_{\sema{\mathcal G''},p}(\sema{\targ' \land \lambda l.\varphi}).$$

The above inclusions imply that
$$\begin{array}{lll}
\sema{\sorc' \land \lambda l.\varphi} 
& = & 
\sema{\sorc'} \cap \sema{\lambda l.\varphi}\\
& \subseteq & 
\mathit{Attr}_{\sema{\mathcal G''},p}(\sema{\targ'}) \cap \sema{\lambda l.\varphi}\\
& \subseteq &
\mathit{Attr}_{\sema{\mathcal G''},p}(\sema{\targ' \land \lambda l.\varphi}).
\end{array}
$$

As $\targ'(\sink_{\mathit{sub}}) = \bot$, we also have that $(\targ'\land  \lambda l.\varphi)(\sink_{\mathit{sub}}) = \bot$.
Hence,  we can use \emph{Property 3} and conclude that
$$\sema{\extend_L(\sorc' \land \lambda l.\varphi)} \subseteq 
\mathit{Attr}_{\sema{\rpg},p}(\sema{\extend_L(\targ' \land \lambda l.\varphi)}).$$

Finally, since the function $\extend_L$ maps to elements of $\symstates$ which assign $\bot$ to each location in $L\setminus L'$,  the last inclusion yields
$$\sema{\extend_L(\sorc') \land \lambda l\in L.\varphi} \subseteq 
\mathit{Attr}_{\sema{\rpg},p}(\sema{\extend_L(\targ') \land \lambda l\in L.\varphi}).$$

This is precisely the statement of the lemma.

We now proceed with establishing the three properties.

\paragraph{Proof of Property 1.}
The proof follows directly from the definition of $\rpg''$ and the fact that the variables in $\cells\setminus\cells'$ that are absent in $\rpg'$ do not appear in the guards or invariants of $\rpg''$ and remain unchanged by the updates in $\rpg''$. 

\paragraph{Proof of Property 2.}
For convenience, let us define $d_\mathbb{Y} := \lambda l.\varphi$.
We prove the claim by transfinite induction,  following the fixpoint definition of attractor. 
More precisely,   for every $d \in \symstates''$ we have that 
$\mathit{Attr}_{\sema{\mathcal G''},p}(\sema{d}) \cap \states''_\env= \mathit{Attr}^\gamma_{\sema{\mathcal G''},p}(\sema{d})$ where $\gamma$ is the smallest ordinal such that
$\mathit{Attr}^\gamma_{\sema{\mathcal G''},p}(\sema{d}) = \mathit{Attr}^{\gamma+1}_{\sema{\mathcal G''},p}(\sema{d})$, where 
\begin{itemize}
\item $\mathit{Attr}^0_{\sema{\mathcal G''},p}(\sema{d}) = \sema{d}$,
\item $\mathit{Attr}^{\beta+1}_{\sema{\mathcal G''},p}(\sema{d}) = \mathit{Attr}^{\beta}_{\sema{\mathcal G''},p}(\sema{d}) \cup \cpre{\sema{\mathcal G''},p}{\mathit{Attr}^{\beta}_{\sema{\mathcal G''},p}(\sema{d})}$,
\item $\mathit{Attr}^{\alpha}_{\sema{\mathcal G''},p}(\sema{d}) = \bigcup_{\beta < \alpha} \mathit{Attr}^\beta_{\sema{\mathcal G''},p}(\sema{d})$ for limit ordinals $\alpha$.
\end{itemize}

We prove that for every $d \in \symstates''$,  
for every ordinal $\alpha$ it holds that
$\mathit{Attr}^\alpha_{\sema{\mathcal G''},p}(\sema{d}) \cap \sema{d_\mathbb{Y}}\subseteq 
\mathit{Attr}^\alpha_{\sema{\mathcal G''},p}(\sema{d \land d_\mathbb{Y}})$, from which the claim follows.

\smallskip
\noindent
\textbf{Case $\alpha=0$.}  
 We have that 
 $\mathit{Attr}^0_{\sema{\mathcal G''},p}(\sema{d}) = \sema{d}$ and 
 $\mathit{Attr}^0_{\sema{\mathcal G''},p}(\sema{d \land d_\mathbb{Y}}) = \sema{d \land d_\mathbb{Y}} = \sema{d} \cap \sema{d_\mathbb{Y}}$,  
 which proves the claim for $\alpha=0$.
 
\smallskip
\noindent
\textbf{Case $\alpha=\beta+1$ is a successor ordinal.}
By the definition of attractor we have that 
$\mathit{Attr}^{\beta+1}_{\sema{\mathcal G''},p}(\sema{d}) = 
\mathit{Attr}^{\beta}_{\sema{\mathcal G''},p}(\sema{d})\cup 
\cpre{\sema{\mathcal G''},p}{\mathit{Attr}^{\beta}_{\sema{\mathcal G''},p}(\sema{d})}$ and 
$\mathit{Attr}^{\beta+1}_{\sema{\mathcal G''},p}(\sema{d \land d_\mathbb{Y}}) = 
\mathit{Attr}^{\beta}_{\sema{\mathcal G''},p}(\sema{d \land d_\mathbb{Y}})\cup 
\cpre{\sema{\mathcal G''},p}{\mathit{Attr}^{\beta}_{\sema{\mathcal G''},p}(\sema{d\land d_\mathbb{Y}})}$.

By induction hypothesis we have
$\mathit{Attr}^\beta_{\sema{\mathcal G''},p}(\sema{d}) \cap \sema{d_\mathbb{Y}}\subseteq 
\mathit{Attr}^\beta_{\sema{\mathcal G''},p}(\sema{d \land d_\mathbb{Y}})$.  
This, together with the monotonicity of the function $\cpre{\sema{\mathcal G''},p}{\cdot}$ implies that 
$\mathit{Attr}^{\beta+1}_{\sema{\mathcal G''},p}(\sema{d \land d_\mathbb{Y}}) \supseteq
\mathit{Attr}^{\beta}_{\sema{\mathcal G''},p}(\sema{d \land d_\mathbb{Y}})\cup 
\cpre{\sema{\mathcal G''},p}{\mathit{Attr}^{\beta}_{\sema{\mathcal G''},p}(\sema{d})\cap\sema{d_\mathbb{Y}}}$.

Therefore,  to establish the claim it suffices to prove for every $A \subseteq \states''_\env$ that
$$ \cpre{\sema{\mathcal G''},p}{A} \cap \sema{d_\mathbb{Y}} \subseteq
\cpre{\sema{\mathcal G''},p}{A\cap\sema{d_\mathbb{Y}}}.
$$
To do this, we use the fact that the variables in $\mathbb{Y}$ are not updated in $\rpg''$.  
Let $\rho''$ be the transition relation of $\sema{\rpg''}$.  
Since for all $l \in L'$ it holds that $d_\mathbb{Y}(l) = \varphi$ and $\varphi \in \FOL{\mathbb{Y}}$ contains only variables in $\mathbb{Y}$, we have that  
\begin{itemize}
\item for every $((l,\assmt{x}),((l,\assmt{x}),\assmt{i})) \in \rho''$,  $(l,\assmt{x}) \in \sema{d_\mathbb{Y}}$ iff $((l,\assmt{x}),\assmt{i}) \in \sema{d_\mathbb{Y}}$;
\item for every $(((l,\assmt{x}), \assmt{i}),(l',\assmt{x}')) \in \rho''$,   $((l,\assmt{x}),\assmt{i}) \in \sema{d_\mathbb{Y}}$ iff $(l',\assmt{x}') \in \sema{d_\mathbb{Y}}$.
\end{itemize}

Let $A \subseteq \states''_\env$, where $\states''$ is the set of states of $\rpg''$. \\
Let $s = (l,\assmt{x})  \in \cpre{\sema{\mathcal G''},p}{A} \cap \sema{d_\mathbb{Y}}$.  We consider two cases.
\begin{itemize}
\item Case $p = \sys$.  \\
For every transition 
$((l,\assmt{x}),((l,\assmt{x}),\assmt{i})) \in \rho''$ it holds that 
$((l,\assmt{x}),\assmt{i}) \in \sema{d_\mathbb{Y}}$  and 
there exists a transition 
$(((l,\assmt{x}), \assmt{i}),(l',\assmt{x}')) \in \rho''$ such  that 
$(l',\assmt{x}') \in A$.  
We also have $(l',\assmt{x}')) \in \sema{d_\mathbb{Y}}$.\\
By the definition of $\cpre{\sema{\mathcal G''},\sys}{\cdot}$, 
we  have  $s \in \cpre{\sema{\mathcal G''},\sys}{A\cap\sema{d_\mathbb{Y}}}$.
\item Case $p = \env$.\\
There exists a transition 
$((l,\assmt{x}),((l,\assmt{x}),\assmt{i})) \in \rho''$ such that 
$((l,\assmt{x}),\assmt{i}) \in \sema{d_\mathbb{Y}}$ and for every
transition 
$(((l,\assmt{x}), \assmt{i}),(l',\assmt{x}')) \in \rho''$ it holds that 
$(l',\assmt{x}') \in A$ and $(l',\assmt{x}')) \in \sema{d_\mathbb{Y}}$.\\
By the definition of $\cpre{\sema{\mathcal G''},\env}{\cdot}$, 
we  have  $s \in \cpre{\sema{\mathcal G''},\env}{A\cap\sema{d_\mathbb{Y}}}$.
\end{itemize}

\smallskip
\noindent
\textbf{Case $\alpha$ is a limit ordinal.}
By induction hypothesis we have that the statement holds for all ordinals $\beta < \alpha$.  Thus, 
$\bigcup_{\beta < \alpha} \Big(\mathit{Attr}^\beta_{\sema{\mathcal G''},p}(\sema{d}) \cap \sema{d_\mathbb{Y}}\Big)\subseteq 
\bigcup_{\beta < \alpha}\mathit{Attr}^\beta_{\sema{\mathcal G''},p}(\sema{d \land d_\mathbb{Y}}).$
By definition we have 
$\Big(\bigcup_{\beta < \alpha} \mathit{Attr}^\beta_{\sema{\mathcal G''},p}(\sema{d})\Big) \cap \sema{d_\mathbb{Y}}=\mathit{Attr}^\alpha_{\sema{\mathcal G''},p}(\sema{d}) \cap \sema{d_\mathbb{Y}}$ and 
$\bigcup_{\beta < \alpha}\mathit{Attr}^\beta_{\sema{\mathcal G''},p}(\sema{d \land d_\mathbb{Y}}) = \mathit{Attr}^\alpha_{\sema{\mathcal G''},p}(\sema{d \land d_\mathbb{Y}}$.
Therefore, the claim for $\alpha$ directly follows from the inclusion above.

\paragraph{Proof of Property 3.}
Let $s = (l,\assmt{x}) \in \sema{\extend_L(\sorc)}$.
By the definition of $\extend_L$  we have that if $l \not\in L'$ then $\extend_L(\sorc)(l) = \bot$, which is not possible since $s \in \sema{\extend_L(\sorc)} $.   
Thus $l \in L\cap L'$, and therefore $s \in \states''$, where $\states''$ is the set of states of $\rpg''$.
Furthermore,  we have that $\assmt{x} \FOLentailsT{T} \sorc(l)$.

Thus, since $\sema{\sorc} \subseteq \mathit{Attr}_{\sema{\rpg''},p}(\sema{\targ})$, 
we have that $s \in \mathit{Attr}_{\sema{\rpg''},p}(\sema{\targ})$.
By the properties of attractor this means that there exists strategy $\sigma''$ of player $p$ in $\sema{\rpg''}$, such that 
every play $\pi \in \plays_{\sema{\rpg''}}(s,\sigma'')$ has a prefix that reaches $\sema{targ}$.
Since $\targ(\sink_{\mathit{sub}}) = \bot$ and $\sink_\mathit{sub}$ is a sink location,  we have that
every play $\pi \in \plays_{\sema{\rpg''}}(s,\sigma'')$ has a prefix that reaches $\sema{targ}$ before visiting $\sink_{\mathit{sub}}$.  
Thus,  we can define a strategy $\sigma$ in $\sema{\rpg}$ that mimics $\sigma''$ as follows:
\begin{itemize}
\item If $\pi \in \states^* \cdot \states_p$ contains a state in $\states\setminus\states''$, then $\sigma(\pi)$ is fixed arbitrarily.
\item If $\pi \in (\states'')^* \cdot \states''_p$ and $\sigma''(\pi) \in \states$, then $\sigma(\pi):=\sigma''(\pi)$.
\item If $\pi \in (\states'')^*  \cdot \states''_p$ and $\sigma''(\pi) \not\in \states$, then $\sigma(\pi)$ is fixed arbitrarily.
\end{itemize}
By the choice of $\sigma''$ and the definition of $\sigma$, we have that every play $\pi \in \plays_{\sema{\rpg}}(s,\sigma)$ has a prefix that reaches $\extend_L(\sema{targ})$.
This implies that  $s \in \mathit{Attr}_{\sema{\rpg},p}(\sema{\extend_L(\targ)})$,  which concludes the proof of \emph{Property 3}.  \qed
\end{proof}

\restateCorrectnessSubgameCache*
\begin{proof}
\textsc{SubgameCache}($\rpg$,  $p$, $L_\mathit{sub}$, $d$) returns a singleton set for the form\\
$\{(\rpg,p,\extend_L(a),\extend_L(d'), \cells_\independent)\}$.
Thus,  we have to show that the tuple $(\rpg,p,\extend_L(a),\extend_L(d'), \cells_\independent)$ satisfies the condition in \Cref{defn:cache}. 

 By line~\ref{line:subgame-attractor} of \cref{alg:subgame-cache}, we have that 
$a:= \textsc{AttractorAcc}(\rpg', p,d')$,  where 
\begin{itemize}
\item $\rpg'$ is a sub-game structure of $\rpg$ induced by 
$L_\mathit{sub}$ and
\item $d' = \lambda l.~\texttt{if } l \in L_\mathit{sub}\texttt{ then } \mathsf{QElim} (\exists (\cells \setminus \cells'). d(l)) \texttt{ else } \bot$.
\end{itemize}
By the soundness of \textsc{AttractorAcc}~\cite{HeimD24},  we have 
$\sema{a} \subseteq \mathit{Attr}_{\sema{\mathcal G'},p}(\sema{d'})$.
As $d'(\sink_\mathit{sub}) = \bot$ for the sink location of $\rpg'$,  and 
$\cells_\independent = \indvars(\rpg,\rpg')$,  then $a$,  $d'$ and $\cells_\independent$ satisfy the preconditions of \Cref{lemma:subgame-cache}, meaning that we can apply it.
Hence, by  \Cref{lemma:subgame-cache} we have that $(\rpg,p,\extend_L(a),\extend_L(d'), \cells_\independent)$ satisfies the condition in \Cref{defn:cache}.  This concludes the proof.
\qed
\end{proof}


\restateCorrectnessAbstraction*
\begin{proof}We prove each of the statements separately.

\smallskip
\noindent
\textit{Proof of  (1).}
Let $\widehat W = W_{\sys}(\widehat G^\uparrow,\widehat \Omega)$, let $s_0 \in W_{\sys}(\sema{\rpg},\Omega)$, and let $v_0 = \abstracts (s_0)$.
To show that $s_0 \in \concretize(\widehat W)$, we will show that $v_0 \in \widehat W$.
To this end, we will define a strategy $\widehat \sigma_\sys$ for player~$\sys$ in $\widehat G^\uparrow$ such that  every play $\widehat \pi \in \plays_{\widehat G^\uparrow}(v_0,\widehat\sigma_\sys)$ is winning for player~$\sys$.

First observe that by definition of $\rho^\uparrow$, the following hold for every states $s\in\states$ and $v\in V$ with $\abstracts(s) = v$: 
(i) for every $(s,s')\in \rho \cap (\states_\sys\times\states_\env)$, we have $(v,\abstracts(s'))\in \widehat\rho^\uparrow \cap (V_\sys \times  V_\env)$;
(ii) for every $(v,v')\in \widehat\rho^\uparrow \cap (V_\env \times  V_\sys)$, there exists $s'\in \states_\sys$ with $\abstracts(s') = v'$ and $(s,s')\in \rho \cap (\states_\env\times\states_\sys)$.
Furthermore, note that there is no dead-ends in $\sema{\rpg}$ by definition. Hence, by (i), there no dead-end in $V_\sys$ in $\widehat G^\uparrow$.
Moreover, as $s_0 \in W_{\sys}(\sema{\rpg},\Omega)$, there exists a strategy $\sigma_\sys$ for player $\sys$ in $\sema{\rpg}$ which is a total function s.t.\ $\plays_{\sema{\rpg}}(s_0,\sigma_\sys)\subseteq\Omega$. 

With this, we define $\widehat \sigma_\sys: V^*\cdot V_\sys \to V$ as follows.
Let $\widehat\tau \in V^*\cdot V_\sys$.  
If there exists no play in $\plays_{\sema{\rpg}}(s_0,\sigma_\sys)$ with a prefix $\tau \in \concretize (\widehat \tau)$, then we set $\widehat \sigma_\sys(\widehat \tau)$ to an arbitrary successor of $\last(\widehat\tau)$.  
Otherwise,  fix one such $\tau$,  and let  $s' := \sigma_\sys(\tau)$.  
We define $\widehat\sigma_\sys(\widehat \tau) = \abstracts(s')$, which is a successor of $\last(\widehat\tau)$ by (i).

Now, let $\widehat \pi\in\plays_{\widehat G^\uparrow}(v_0,\widehat\sigma_\sys)$.
If $\widehat\pi$ is an infinte play,
using (ii) and the definition of $\widehat\sigma_\sys$, we can inductively show that there exists a play $\pi\in\plays_{\sema{\rpg}}(s_0,\sigma_\sys)$ s.t.\ $\abstracts(\pi) = \widehat \pi$.
Then, by definition, $\loc(\widehat\pi) = \loc(\pi)$. 
As $\pi\in\Omega$, we have $\widehat\pi\in \abstracts(\Omega) = \widehat\Omega$.
Alternatively, if $\widehat\pi$ is finite, $\last(\widehat\pi)\in V_\env$ as there is no dead-ends in $V_\sys$.
Furthermore, for safety objectives, we can show that there exists a play in $\plays_{\sema{\rpg}}(s_0,\sigma_\sys)$ with prefix $\pi$ s.t.\ $\abstracts(\pi) = \widehat \pi$. 
As $\pi$ is a prefix of a play in $\Omega$, $\widehat\pi$ is also a prefix of a play in $\widehat\Omega$.
Hence, in any case, $\widehat\pi$ is winning for player $\sys$.\\

\smallskip
\noindent
\textit{Proof of  (2).}
Let $\widehat W = W_{\sys}(\widehat G^\downarrow,\widehat \Omega)$, $v_0 \in \widehat W$, and let $s_0\in\concretize(v_0)$.
We will show that $s_0 \in W_{\sys}(\sema{\rpg},\Omega)$.
To this end, we will define a strategy $\sigma_\sys$ for player~$\sys$ in $\sema{\rpg}$ such that every $\pi \in \plays_{\sema{\rpg}}(s_0,\sigma_\sys)$ is winning for player $\sys$. 

First observe that by definition of $\rho^\downarrow$, the following hold for every states $s\in\states$ and $v\in V$ with $\abstracts(s) = v$: 
(i) for every $(s,s')\in \rho \cap (\states_\env\times\states_\sys)$, we have $(v,\abstracts(s'))\in \widehat\rho^\downarrow \cap (V_\env \times  V_\sys)$;
(ii) for every $(v,v')\in \widehat\rho^\downarrow \cap (V_\sys \times  V_\env)$, there exists $s'\in \states_\env$ with $\abstracts(s') = v'$ and $(s,s')\in \rho \cap (\states_\sys\times\states_\env)$.
Furthermore, as $v_0 \in \widehat W$, there exists a strategy $\widehat \sigma_\sys$ for player~$\sys$ in $\widehat G^\downarrow$ such that every $\widehat \pi \in \plays_{\widehat G^\downarrow}(v_0,\widehat\sigma_\sys)$ is winning for player $\sys$.

With this, we define $\sigma_\sys: \states^*\cdot \states_\sys \to \states$ as follows.
Let $\tau\in \states^*\cdot \states_\sys$.
If $\widehat \sigma_\sys$ is not defined for $\abstracts(\tau)$, then we set $\sigma_\sys(\tau)$ to an arbitrary successor of $\last(\tau)$. 
Otherwise, for $\widehat\sigma_\sys (\abstracts(\tau)) = v$, by (ii), there exists $s$ with $\abstracts(s) = v$ and $(\last(\tau),s)\in \rho \cap (\states_\sys\times\states_\env)$.
For such cases, fix one such $s$ and set $\sigma_\sys(\tau) = s$.

Now, let $\pi\in \plays_{\sema{\rpg}}(s_0,\sigma_\sys)$.
As there is no dead-end in $\sema{\rpg}$, $\pi$ is an infinite play.
Then, using (i) and the definition of $\sigma_\sys$, we can inductively show $\widehat\pi = \abstracts(\pi)$ is an infinite play in $\plays_{\widehat G^\downarrow}(v_0,\widehat\sigma_\sys)$.
Then, by definition, $\loc(\widehat\pi) = \loc(\pi)$. 
As $\widehat\pi\in\widehat\Omega$, we have $\pi\in\Omega$, and hence, $\pi$ is winning for player $\sys$.\qed
\end{proof}

%% file: appendix-pruning.tex
Let $p \in \{\sys,\env\}$ be a player and $d \in \symstates$.
We now define a reactive program game $\prune(\rpg,\Omega, d, p)$ obtained from $(\rpg,\Omega)$ by redirecting all transitions from the states in $\sema{d}$ to a sink location for player $p$.
Intuitively, we \emph{prune} the existing transitions originating in states in $d$.  We do this in two steps.

First, we augment the reactive program game with two \emph{sink locations} $\sink_\sys$ and $\sink_\env$, one for each player,  and modify the objective $\Omega$ such that, intuitively,  states with location $\sink_p$ will be losing for player $p$. 
We define $(\rpg_\sink,\Omega_\sink)$ as follows. 
Let $\rpg_\sink := (T, \inputs, \cells, L\uplus\{\sink_\sys,\sink_\env\}, \Inv', \delta')$ where
\begin{itemize}
\item $\Inv'(l) := \Inv(l)$ if $l \in L$ and $\Inv'(l) = \top$ otherwise, and
\item $\delta' := \delta \uplus \{(\sink_\sys,\top,\lambda x.\;x,\sink_\sys),(\sink_\env,\top,\lambda x.\;x,\sink_\env)\}$.
\end{itemize}
Let $\states'$ be the set of states of $\rpg_\sink$.  We extend the objective $\Omega$ to $\Omega_\sink$ for player $\sys$ based on the different possible types of objectives $\Omega$.

\smallskip
\noindent
\textbf{Case} $\Omega = \mathit{Safety}(S)$ for some $S \subseteq L$.
$$\Omega_\sink := \Omega \cup \{\pi\in {\states'}^\omega \mid \exists n \in \Nat. ~\loc(\pi[n])=\sink_\env\land \forall m <n. ~\loc(\pi[m])\in S\}.$$
Here, the new objective also includes the plays that reach the location $\sink_\env$ while staying within the safe set $S$.

\smallskip
\noindent
\textbf{Case} $\Omega = \mathit{Reach}(R)$ for some $R \subseteq L$.
$$\Omega_\sink := \Omega \cup \{\pi\in {\states'}^\omega \mid \exists n \in \Nat. ~\loc(\pi[n])\in\{\sink_\env\} \cup R \land \forall m <n. ~\loc(\pi[m])\in L\}.$$
Here, the new objective includes all the plays that reach $\sink_\env$ or $R$ without visiting $\sink_\sys$.

\smallskip
\noindent
\textbf{Case} $\Omega$ is prefix-independent.
$$\Omega_\sink := \Omega \cup \{\pi\in {\states'}^\omega \mid \exists n \in \Nat. ~\loc(\pi[n])=\sink_\env \land \forall m <n. ~\loc(\pi[m])\in L\}.$$
Here, the new objective also includes the plays that reach location $\sink_\env$ without visiting $\sink_\sys$.

Now,  we construct the reactive program game $\prune(\rpg, \Omega,d, p)$ from $(\rpg_{\sink},\Omega_\sink)$ by redirecting all transitions from $d$ to the location $\sink_p$. Formally,  we let
$\prune(\rpg, \Omega,d, p) := ((T, \inputs, \cells, L\uplus\{\sink_\sys,\sink_\env\}, \Inv', \delta''),\Omega_\sink)$ where 
    \[ \delta'' := \{ (l, g \land \lnot d(l), u, l')\mid (l, g, u, l') \in \delta \} \cup  \{ (l,  d(l),\lambda x.\;x, \sink_{p})\mid l \in L, d(l) \not\equiv_{T}\bot\} . 
    \]

The following lemma formalizes the desired property of the pruning: Pruning states that are losing for player $p$ does not 
change the winning region.

\begin{lemma}[Correctness of Pruning]\label{lemma:pruning}
If $\sema{d} \subseteq W_{1-p}(\rpg, \Omega)$,  then
    \[ W_p(\rpg, \Omega) = W_p(\prune(\rpg,\Omega, d, p))\cap \states.\]
\end{lemma}
\begin{proof}
	Let $\prune(\rpg,\Omega, d, p) = (\rpg',\Omega')$ with 
	semantics 
	$\sema{\rpg'} = (\states',\states_\env',\states_\sys',\rho')$.
	Furthermore, let $W = W_p(\rpg, \Omega)$ and $W' = W_p(\rpg', \Omega')$.
	As $\sema{\rpg}$ (resp. $\sema{\rpg'}$) have no dead-end, a play $\pi$ is winning for Player $\sys$ iff $\pi\in\Omega$ (resp. $\pi\in\Omega'$).
	Hence, let us denote the objectives for $\sys$ and $\env$ as $\Omega_\sys = \Omega$ and $\Omega_\env = \states^\omega\setminus\Omega$, respectively in $\rpg$. $\Omega_\sys'$ and $\Omega_\env'$ are defined analogously.
	Now, we will show both direction of $W = W'\cap \states$ separately.

	\smallskip
	\noindent $(\subseteq)$ 
	Let $s_0\in W$, then there exists a strategy $\sigma_p$ for Player $p$ in $\sema{\rpg}$ such that $\plays_{\sema{\rpg}}(s_0,\sigma_p)\subseteq \Omega_p$.
	To show that $s_0\in W'$, we will define a strategy $\sigma_p'$ for Player $p$ in $\sema{\rpg'}$ such that $\plays_{\sema{\rpg'}}(s_0,\sigma_p')\subseteq \Omega_p'$.

	First observe that for every $s = ((l, \assmt{x}),\assmt{i})\in W\cap \states_\sys$, we have $s\not\in\sema{d}$ and hence, $\assmt{x}\uplus\assmt{i}\FOLentailsT{T} \neg d(l)$. 
	So, it holds that
	\begin{align*}
		&(((l,\assmt{x}), \assmt{i}),(l',\assmt{x}')) \in \rho \\
		&\iff \exists (g,u) \in \Labels(l,l').~ \assmt{x}\uplus\assmt{i} \FOLentailsT{T} g, \assmt{x}'(x) =  \eval{\assmt{x}\uplus\assmt{i}}(u(x))\\
		&\iff \exists (g,u) \in \Labels(l,l').~ \assmt{x}\uplus\assmt{i} \FOLentailsT{T} g\wedge\neg d(l), \assmt{x}'(x) =  \eval{\assmt{x}\uplus\assmt{i}}(u(x))\\
		&\iff (((l,\assmt{x}), \assmt{i}),(l',\assmt{x}')) \in \rho'.
	\end{align*}
	Furthermore, by definition, for every $s\in W\cap\states_\env$, $(s,(s,\assmt{i})) \in \rho$ iff $(s,(s,\assmt{i})) \in \rho'$.
	Therefore, $\rho \cap (W\times \states) = \rho' \cap (W\times \states)$.

	With this, we define $\sigma_p': (\states')^*\cdot \states_p' \to \states$ as follows.
	Let $\tau\in (\states')^*\cdot \states_p'$.
	If $\tau\in W^*$, then we define $\sigma_p'(\tau) = \sigma_p(\tau)$, else we set $\sigma_p'(\tau)$ to an arbitrary successor of $\last(\tau)$.

	Now, let $\pi\in\plays_{\sema{\rpg'}}(s_0,\sigma_p')$.
	By definition, $\pi$ is an infinite play.
	If $\pi\in W^\omega$, then, by definition, $\pi\in\plays_{\sema{\rpg}}(s_0,\sigma_p)$ and hence, $\pi\in\Omega_p\subseteq\Omega_p'$.

	Now, assume $\pi\not\in W^\omega$, then there exists $k\in\mathbb{N}$ s.t.\ $\pi[k]\not\in W$ and $\pi[j]\in W$ for all $j < k$.
	As $\rho \cap (W\times \states) = \rho' \cap (W\times \states)$, $\pi[0,k]$ is also a prefix of a play in $\plays_{\sema{\rpg}}(s_0,\sigma_p)$.
	Since $\pi[k]\not\in W$, Player $1-p$ has a strategy $\sigma_{1-p}$ such that $\plays_{\sema{\rpg}}(\pi[k],\sigma_{1-p})\subseteq \Omega_{1-p}$. 
	For safety objective or any prefix-independent objective $\Omega_p$,
	if Player $p$ uses a strategy that is consistent with $\pi$ until $\pi[k]$ and then switches to $\sigma_{1-p}$, then this gives us an infinite play $\pi'\in \plays_{\sema{\rpg}}(s_0,\sigma_p)\cap \Omega_{1-p}$. This is a contradiction to the assumption that $\plays_{\sema{\rpg}}(s_0,\sigma_p)\subseteq\Omega_p$.
	Similarly, for reachability objective $\Omega_p = \reach(R)$, if $\pi[j]\not\in R$ for all $j\leq k$, then the same argument as above gives us a contradiction.
	If there exists $j\leq k$ s.t.\ $\pi[j]\in R$, then $\pi\in\Omega_p'$ by definition.

	\smallskip
	\noindent $(\supseteq)$ 
	Let $s_0\in W'\cap\states$, then there exists a strategy $\sigma_p'$ for Player $p$ in $\sema{\rpg'}$ such that $\plays_{\sema{\rpg'}}(s_0,\sigma_p')\subseteq \Omega_p'$.
	To show that $s_0\in W$, we will define a strategy $\sigma_p$ for Player $p$ in $\sema{\rpg}$ such that $\plays_{\sema{\rpg}}(s_0,\sigma_p)\subseteq \Omega_p$.

	We define $\sigma_p: \states^*\cdot \states_p \to \states$ as follows.
	Let $\tau\in \states^*\cdot \states_p$.
	If $\tau\in \states^*$ and $\sigma_p'(\tau)\in\states$, then we define $\sigma_p(\tau) = \sigma_p'(\tau)$, else we set $\sigma_p(\tau)$ to an arbitrary successor of $\last(\tau)$.

	Now, let $\pi\in\plays_{\sema{\rpg}}(s_0,\sigma_p)$.
	If $\pi\in\plays_{\sema{\rpg'}}(s_0,\sigma_p')$, then $\pi\in\Omega_p'\cap \states^\omega = \Omega_p$.
	Otherwise, there exists prefix $\tau$ of $\pi$ s.t.\ $\sigma_p'(\tau) \in \{\sink_p,\sink_{1-p}\}$ and $\tau\sigma_p'(\tau)$ is a prefix of a play in $\plays_{\sema{\rpg'}}(s_0,\sigma_p')$.
	By construction, there is no transition from $\states\cup\{\sink_p\}$ to $\sink_{1-p}$, hence, $\sigma_p'(\tau) = \sink_p$.
	Furthermore, by construction, there is no play in $\Omega_p'$ that visits state $\sink_p$.
	This leads to a contradiction to the assumption that $\plays_{\sema{\rpg'}}(s_0,\sigma_p')\subseteq\Omega_p'$.\qed
\end{proof}

The next statement follows directly from \Cref{lemma:pruning} and \Cref{lemma:abstraction} and allows us to soundly prune the states determined to be losing for the over-approximated player $\overapproxp(\widehat{G}^\circ)$  in an abstract game $\widehat{G}^\circ$.
 
\begin{corollary}\label{corollary:pruning}
Let $\widehat{G}^\circ$ for some $\circ\in\{\uparrow,\downarrow\}$ be a 
$(\predss,\predsi)$-induced abstraction of  $\rpg$ and let
$p = \overapproxp(\widehat{G}^\circ)$.
If $d \in \symstates$ is such that 
$\sema{d} \subseteq \concretize(W_{1-p}(\widehat{G}^\circ,\widehat \Omega))$, 
then 
$W_p(\rpg, \Omega) = W_p(\prune(\rpg,\Omega, d, p))\cap \states$.
\end{corollary}

\Cref{alg:rpg-enhanced-prune} shows the version \textsc{RPGPruneCacheSolve} of \Cref{alg:rpg-enhanced} extended with pruning.
\textsc{SolveAbstractWR} works in the same way as \textsc{SolveAbstract}, but additionally returns the winning region for the opponent.
The correctness of \Cref{alg:rpg-enhanced-prune} follows from \Cref{corollary:pruning} and \Cref{thm:method-correctness}.

\begin{theorem}[Correctness of \Cref{alg:rpg-enhanced-prune}] \label{thm:prune-cache-method-correctness}
Given a reactive program game structure $\rpg$ and a location-based objective $\Omega$,  for any $b \in \Nat$, if  \textsc{RPGPruneCacheSolve} terminates, then it returns $W_{\sys}(\sema{\rpg},\Omega).$
\end{theorem}

\begin{algorithm}[t!]
\SetAlgoVlined
    \SetKwProg{Fn}{function}{}{}
    \DontPrintSemicolon
    \Fn{\textsc{RPGPruneCacheSolve}(
		 $\mathcal G = (T,\inputs, \cells, L, \Inv,\delta)$,  
    	 $\Omega$, $b \in \Nat$)}{
	     $(\predss,\predsi):= \mathsf{AbstractDomain}(\rpg)$\;
		 $(\widehat G^\uparrow,\widehat G^\downarrow) := \mathsf{AbstractRPG}(\rpg,(\predss,\predsi))$\;
		 $(\safegroup_\sys,\colivegroup_\sys,\livegroup_\sys,\widehat W_\env) :=\textsc{SolveAbstractWR}(\widehat G^\uparrow,\Omega)$\;
		 $(\safegroup_\env,\colivegroup_\env,\livegroup_\env,\widehat W_\sys) :=\textsc{SolveAbstractWR}(\widehat G^\downarrow,\Omega)$\;		
		 $(\rpg',\Omega') := \prune(\rpg,\Omega,\widehat W_\env,\sys)$\;
		 $(\rpg'',\Omega'') := \prune(\rpg',\Omega',\widehat W_\sys,\env)$\;
		 $C_\sys := \textsc{GenerateCache}(\rpg'', \widehat{G}^\uparrow, \sys, (\safegroup_\sys,\colivegroup_\sys,\livegroup_\sys), b )$\;
		 $C_\env := \textsc{GenerateCache}(\rpg'', \widehat{G}^\downarrow, \env, (\safegroup_\env,\colivegroup_\env,\livegroup_\env), b)$\;				
\Return \textsc{RPGSolveWithCache}($\rpg'', C_\sys \cup C_\env$)
	}
\caption{Procedure for solving reactive program games enhanced with abstraction-based pruning and abstract template-based caching.}
\label{alg:rpg-enhanced-prune}
\end{algorithm}

%% file: appendix-benchmarks.tex
All our benchmarks are modeled as reactive program games with Büchi objectives for the system player. We describe all our benchmarks in detail below.

\paragraph*{Scheduler.}
This benchmark outlines two primary tasks for the system: a global task and one repeated task. Initially, three program variables, namely $taskG$, $taskR$, and $boundG$, are set to specific values determined by the environment, where the system can choose to set $taskR$ as the negation of the value specified by the environment.
The system undergoes four distinct stages. At each stage, if $taskR > 0$, it proceeds to the next stage; otherwise, it has the option to increment either of the variables $taskG$ or $taskR$.
For the repeated task,  the system must attain the state $taskR > 0$ in every stage, whereas for the global task, the system must eventually ensure $taskG > boundG$. 
Crucially, the system must strategically ensure that there is at least one stage where $taskR \leq 0$. This specific condition allows the system to increment the variable $taskG$, thereby progressing toward fulfilling the global task. The entire process resets after the completion of the four stages.

\paragraph*{Item Processing.}
In this benchmark, the system processes a tray of items, manipulating variables such as $trayItems$ and $numItems\_s$ through different locations. Similar to the scheduler benchmark, these variables are initially set to the values determined by the environment. At each location, the system is presented with various operations it can perform on these variables. The primary goal of the system is to guarantee the condition $numItems\_s \geq trayItems$ in the $done$ location.
To meet this objective, the system must strategically select the appropriate options at each step.

\paragraph*{Chain Benchmarks.}
We have two categories of chain benchmarks, both designed with the objective $\buchi(\{goal\})$.
\begin{itemize}
    \item In \textbf{chain $k$}, we parameterize the number of locations and program variables based on the value $k$.
    Specifically, the benchmark has $k+3$ locations $\{int, goal, sink, l_1,l_2,\ldots,l_k\}$, $k+2$ program variables $\{y,c,x_1,x_2,\ldots,x_k\}$, and one input variable $i$. Starting from initial location $int$, it sets $c=0$ and goes to location $goal$. 
    From $goal$, if $c>0$, it goes to the $sink$ which only allows a self-loop. If $c\leq 0$, it sets $x_1 = i$, $y = i$ and goes to location $l_1$.
    From every location $l_j$ with $j < k$, if $x_j = 0$, then it can go to next location $l_{j+1}$ by setting $x_{j+1} = y$. 
    Alternatively, if $x_j\neq 0$, it can either loop on $l_j$ while incrementing/decrementing $x_j$ by $1$, or go to $goal$ by setting $c=i$. 
    The transitions from $l_k$ are similar with the exception that it directly goes to $goal$ if $x_k=0$.
    \item In \textbf{chain simple $k$}, we parametrize only the number of locations based on the value $k$.
    Similar to \textbf{chain k}, this benchmark has $k+3$ locations $\{int, goal, sink, l_1,l_2,\ldots,l_k\}$ and one input varibale $i$. However, it only has $3$ program variables $\{y,c,x\}$. The transition relation is similar to that of \textbf{chain k}, but the program variable $x$ is used instead of $x_1, x_2, \ldots, x_k$.
\end{itemize}

\paragraph*{Robot Benchmarks.}
These benchmarks describe a robot's movement along a one-dimensional discrete grid. The robot's position is determined by a program variable, denoted as $\mathit{trackPos}$ (for clarity, we use $\mathit{pos}$ in \Cref{fig:running}).
The robot can move either one step forward, incrementing the variable $\mathit{trackPos}$ by $1$, or one step backward, decrementing $\mathit{trackPos}$ by $1$.
\begin{itemize}
    \item In \textbf{robot analyze samples}, the robot starts from a designated location $\mathit{base}$.
    It then collects a number of samples (as specified by the environment) from another location. 
    Subsequently, it moves to another location $\mathit{lab}$, conducts tests on all collected samples, and finally returns to the $\mathit{base}$.
    \item In \textbf{robot repair}, the robot starts from a designated location $\mathit{base}$.
    If the device is deemed faulty (as determined by the environment) the robot gets a number of spare parts (as specified by the environment) from the stock. It proceeds to repair the device using all the spare parts and eventually returns to the $\mathit{base}$.
    \item In \textbf{robot collect samples}, the robot starts from a designated location $\mathit{base}$, collects a number of samples (as specified by the environment) from another location, and then returns to the $\mathit{base}$.
    \item In \textbf{robot deliver products $k$}, the robot is tasked with purchasing $k$ distinct products for the $\mathit{office}$. It follows an alternating pattern: choosing a product, storing the required quantity (as specified by the environment), proceeding to the $\mathit{bank}$ to withdraw the corresponding amount of money, purchasing the specified products from the $\mathit{store}$, and finally returning to the $\mathit{office}$.
\end{itemize}

\paragraph*{Smart Home Benchmarks.}
These benchmarks describe the dynamic adjustments made by a smart home in response to environmental inputs, involving temperature regulation, blinds adjustment, and lighting control. The smart home turns on/off lights by toggling a boolean program variable $\mathit{light}$ to true/false. The blinds' position is adjusted by incrementing or decrementing the program variable $\mathit{blinds}$ by $0.1$. Additionally, the temperature level is adjusted by incrementing or decrementing the program variable $\mathit{temperature}$ by $1.0$, influenced by an input variable $\mathit{disturbance}$.
\begin{itemize}
    \item In \textbf{smart home day not empty}, during the day and when the home is occupied (as determined by the environment), the smart home activates $\mathit{daymode}$, turns on the lights, raises the blinds, and raises the temperature to a range specified by the program variable $\mathit{minimum}$.
    \item In \textbf{smart home day warm}, the smart home also performs the tasks outlined in \textbf{smart home day not empty}. Additionally, if the environment indicates the home being too warm, it decrements the variable $\mathit{minimum}$ by $2.0$ and continues its routine.
    \item In \textbf{smart home day cold}, similar to \textbf{smart home day warm}, if the environment indicates the home being too cold, it increments the variable $\mathit{minimum}$ by $2.0$ and continues its routine.
    \item In \textbf{smart home day warm or cold}, the smart home combines the tasks from both \textbf{smart home day warm} and \textbf{smart home day cold}.
    \item In \textbf{smart home empty}, during the day and when the home is empty (as determined by the environment), the smart home activates $\mathit{daymode}$, turns off the lights, lowers the blinds, and lowers the temperature to a range specified by the program variable $\mathit{maximum}$.
    \item In \textbf{smart home night sleeping}, during the night, the smart home deactivates $\mathit{daymode}$ and lowers blinds. If the owner is sleeping (as determined by the environment), it turns off lights; otherwise, it turns them on. Subsequently, it raises the temperature to a range specified by the variable $\mathit{minimum}$.
    \item In \textbf{smart home night empty}, during the night, the smart home deactivates $\mathit{daymode}$ and lowers blinds. If the home is empty (as determined by the environment), it turns off the lights and lowers the temperature to a range specified by the program variable $\mathit{maximum}$.
    \item In \textbf{smart home nightmode}, the smart home combines the tasks from both \textbf{smart home night sleeping} and \textbf{smart home night empty}.
\end{itemize}